%% file: main.tex
\documentclass[runningheads]{llncs}
\usepackage{graphicx}
\input{macro.tex}

\begin{document}
\title{Implicit Enumeration of Topological-Minor-Embeddings and Its Application to Planar Subgraph Enumeration\thanks{This work was supported by JSPS KAKENHI Grant Numbers JP15H05711, JP18H04091, JP18K04610, JP18K11153, and JP19J21000.}}
\titlerunning{Implicit Enumeration of Topological-Minor-Embeddings}
\author{Yu Nakahata\inst{1}\orcidID{0000-0002-8947-0994} \and
Jun Kawahara\inst{1}\orcidID{0000-0001-7208-044X} \and
Takashi Horiyama\inst{2} \and
Shin-ichi Minato\inst{1}\orcidID{0000-0002-1397-1020}}
\authorrunning{Y. Nakahata et al.}
\institute{Kyoto University, Japan\\
\email{\{nakahata.yu.27e@st, jkawahara@i, minato@i\}.kyoto-u.ac.jp} \and
Hokkaido University, Japan\\
\email{horiyama@ist.hokudai.ac.jp}}
\maketitle              % typeset the header of the contribution
\begin{abstract}
Given graphs $G$ and $H$, we propose a method to implicitly enumerate topological-minor-embeddings of $H$ in $G$ using decision diagrams.
We show a useful application of our method to enumerating subgraphs characterized by forbidden topological minors, that is, planar, outerplanar, series-parallel, and cactus subgraphs.
Computational experiments show that our method can find all planar subgraphs in a given graph at most five orders of magnitude faster than a naive backtracking-based method.

\end{abstract}
\section{Introduction}
Subgraph enumeration is a fundamental task in several areas of computer science.
Many researchers have proposed dedicated algorithms for subgraph
enumeration to solve various problems.
However, the number of subgraphs can be exponentially larger than the size of an input graph.
In such cases, it is unrealistic to enumerate all subgraphs explicitly.
To overcome this difficulty, we focus on \emph{implicit} enumeration algorithms~\cite{kawahara2017frontier,knuth2011art,sekine1995tutte}.
Such an algorithm constructs a \emph{decision diagram} (DD)~\cite{bryant1986graph,minato1993zero} representing the set of subgraphs instead of explicitly enumerating the subgraphs.
DDs are tractable representations of families of sets.
By regarding a subgraph of a graph as the subset of edges, DDs can represent a set of (edge-induced) subgraphs.
The efficiency of implicit enumeration algorithms does not directly depend on the number of subgraphs but rather on the size of the output DD~\cite{kawahara2017frontier}.
The size of a DD can be much (exponentially in some cases) smaller than the number of subgraphs, and thus, in such cases, we can expect that the implicit algorithms will work much faster than explicit ones.
The framework to construct a DD representing a set of constrained subgraphs is called \emph{frontier-based search} (FBS)~\cite{kawahara2017frontier}.
Recently, Kawahara et al.~\cite{kawahara2019colorful} proposed an extension of FBS, \emph{colurful FBS} (CFBS).
Especially, they proposed an algorithm to enumerate isomorphic subgraphs, that is, given graphs $G$ and $H$, it constructs a DD representing the set of subgraphs of $G$ that are isomorphic to $H$.

To extend types of subgraphs that can be implicitly enumerated, we focus on \emph{topological-minor-embeddings} (TM-embeddings)~\cite{diestel2012graph}.
A subgraph $G'$ of $G$ is a TM-embedding of $H$ in $G$ if $G'$ is isomorphic to a \emph{subdivision} of $H$.
A subdivision of $H$ is a graph obtained by replacing each edge in $H$ with a path.
A TM-embedding is a generalization of an isomorphic subgraph, where each edge must be replaced with a path with length one.
An application of TM-embeddings is \emph{forbidden topological minor characterization} (FTM-characterization)~\cite{diestel2012graph}.
For example, a graph is planar if and only if it contains neither $K_{5}$ nor $K_{3, 3}$ as a topological minor~\cite{diestel2012graph}, where $K_{a}$ is the complete graph with $a$ vertices and $K_{b, c}$ is the complete bipartite graph with the two parts of $b$ and $c$ vertices.
Other examples are shown in \tabref{tab:graph_classes} (see~\cite{brandstadt1999graph} for details).

\begin{table}[t]
  \centering
  \caption{Relationship between graph classes and forbidden topological minors. $K_4 - e$ is the graph obtained by removing an arbitrary edge from $K_4$.}
  \begin{tabular}{c|c}
    graph class & forbidden topological minors \\ \hline
    planar graphs & $K_5, K_{3, 3}$ \\
    outerplanar graphs & $K_4, K_{2, 3}$ \\
    series-parallel graphs & $K_4$ \\
    cactus graphs & $K_4 - e$ \\
  \end{tabular}
  \label{tab:graph_classes}
\end{table}

\paragraph*{Our contribution.}
In this paper, when graphs $G$ and $H$ are given, we propose an algorithm to implicitly enumerate all TM-embeddings of $H$ in $G$.
Using our algorithm and some additional DD operations, we can also implicitly enumerate subgraphs having FTM-characterizations, that is, planar, outerplanar, series-parallel, and cactus subgraphs.
Our contributions are:
\begin{itemize}
  \item Given graphs $G$ and $H$, we propose an algorithm to implicitly enumerate all TM-embeddings of $H$ in $G$. (\secref{sec:general})
  \item We propose more efficient algorithms when $H$ is in specific graphs that are used in FTM-characterizations of graph classes in \tabref{tab:graph_classes}, that is, complete graphs, complete bipartite graphs, and $K_4 - e$. (\secref{sec:profile})
  \item Combining our algorithm with DD operations, we show how to implicitly enumerate subgraphs having FTM-characterization, that is, planar, outerplanar, series-parallel, and cactus subgraphs. (\secref{sec:FTM-DD})
  \item We evaluate our method by computational experiments. We apply our method to implicitly enumerating all planar subgraphs in a graph. The results show that our method runs up to five orders of magnitude faster than a naive backtracking-based method. (\secref{sec:experiment})
\end{itemize}

\paragraph*{Our techniques.}
Our key contribution is to extend the algorithm for implicitly enumerating isomorphic subgraphs~\cite{kawahara2019colorful} to TM-embeddings.
We first explain the algorithm of Kawahara et al.~\cite{kawahara2019colorful} briefly.
They reduced isomorphic subgraph enumeration to enumeration of ``colored degree specified subgraphs''.
Let us consider $K_{2,3}$ in \figref{fig:K2x3}.
It has the degree multiset $\set{2^3, 3^2}$, where $2^3$ means that there are three vertices with degree 2.
The graph in \figref{fig:house} has the same degree multiset although it is not isomorphic to $K_{2,3}$.
Thus, the degree multiset is not enough to characterize $K_{2,3}$ uniquely.
Let us consider the edge-colored graph in \figref{fig:K2x3_color}.
For an edge-colored graph with $k$ colors, we consider a \emph{colored degree} of a vertex $v$, that is, a $k$-tuple of integers such that its $i$-th element is the number of color-$i$ edges incident to $v$.
The edge-colored graph in \figref{fig:K2x3_color} has the colored-degree multiset $M = \set{(3, 0), (0, 3), (1, 1)^3}$, where red and green are respectively colors 1 and 2.
Observe that every edge-colored graph with colored-degree multiset $M$ is isomorphic to the graph.
Therefore, enumerating subgraphs of $G$ that are isomorphic to $K_{2,3}$ is equivalent to finding all $2$-colored subgraphs of $G$ whose degree multisets equal $M$ and then ``decolorizing'' them.

In TM-embedding enumeration, the size of subgraphs we want to find is not fixed, which makes difficult to identify the subgraphs only by a colored degree multiset.
Therefore, we extend the algorithm of Kawahara et al.\ in two directions: (a) we allow some colored degrees to exist arbitrary many times and (b) we introduce the constraint that each colored subgraph is connected.
Using such constraints, we can reduce TM-embeddings enumeration to enumerating subgraphs satisfying the constraints.
We propose a framework to enumerate TM-embeddings based on such constraints that can be applied to \emph{every} query graph.
Since the complexity of CFBS heavily depends on the number of colors used in the constraint, we discuss how to reduce the number of colors.

\section{Preliminaries}\label{sec:preliminaries}

\subsection{Graphs and colored graphs}\label{sec:notation}
Let $G = (V(G), E(G))$ be a graph.
We define $n = \size{V(G)}$ and $m = \size{E(G)}$.
Graphs $G$ and $H$ are \emph{isomorphic} if there exists a bijection $\func{f}{V(G)}{V(H)}$ such that, for all $u, v \in V(G)$, $\edge{u}{v} \in E(G) \Leftrightarrow \edge{f(u)}{f(v)} \in E(H)$.
For a positive integer $k$, we define $[k] = \set{1, \dots, k}$. For a positive integer $c$ and a finite set $E$, a \emph{$c$-colored subset} over $E$ is a $c$-tuple $(D_1, \dots, D_c)$ such that (a) for each $i \in [c]$, $D_i \subseteq E$ holds and (b) for all $i, j \in [c]$, if $i \neq j$, then $D_i \cap D_j = \emptyset$.
A \emph{$c$-colored graph} $H^c$ is a tuple of a finite set $U$ and a $c$-colored subset $(D_1, \dots, D_c)$ over the set $\setST{\edge{u}{v}}{u, v \in U, u \neq v}$.
Here, $U$ is the \emph{vertex set} and $(D_1, \dots, D_c)$ is the \emph{colored edge set}.
The \emph{color-$i$ degree} of $u \in U$ in $H^c$ is the number of color-$i$ edges incident to $u$.
The \emph{colored degree} of $u$ in $H^c$ is a $c$-tuple $(\delta_1, \dots, \delta_c)$ of integers, where $\delta_i$ is the color-$i$ degree of $u$.
The \emph{colored degree multiset} of $H^c$, which is denoted by $\coloredDegreeSet{H^c}$, is the multiset of the colored degrees of all the vertices in $H^c$.
$H^c$ is a \emph{$c$-colorized graph} of $H$ if the underlying graph of $H^c$ is isomorphic to $H$, where the \emph{underlying graph} of $H^c$ is a graph obtained by ignoring the colors of the edges in $H^c$.
A \emph{$c$-colored subgraph} of $G$ is a $c$-colored graph whose underlying graph is isomorphic to a subgraph of $G$.

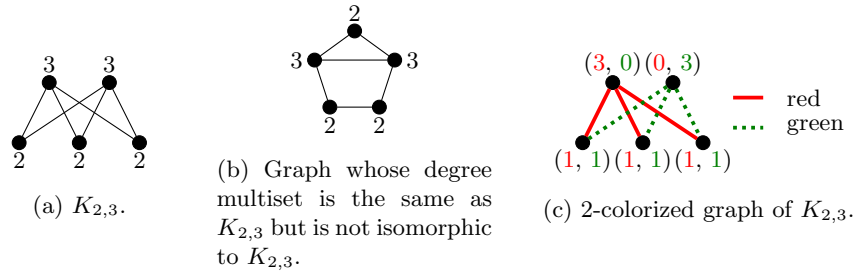
\begin{figure}[t]
  \centering
  \begin{subfigure}{.25\linewidth}
    \centering
    \begin{tikzpicture}[scale=0.8]
      \coordinate (ca) at (0.5, 1) {};
      \coordinate (cb) at (1.5, 1) {};
      \coordinate (cA) at (0.0, 0) {};
      \coordinate (cB) at (1.0, 0) {};
      \coordinate (cC) at (2.0, 0) {};

      \node[GraphNode] (a) at (ca) {};
      \node[GraphNode] (b) at (cb) {};
      \node[GraphNode] (A) at (cA) {};
      \node[GraphNode] (B) at (cB) {};
      \node[GraphNode] (C) at (cC) {};

      \node at ($(ca)+( 0.0,  0.3)$) {3};
      \node at ($(cb)+( 0.0,  0.3)$) {3};
      \node at ($(cA)+( 0.0, -0.3)$) {2};
      \node at ($(cB)+( 0.0, -0.3)$) {2};
      \node at ($(cC)+( 0.0, -0.3)$) {2};

      \draw (a) -- (A);
      \draw (a) -- (B);
      \draw (a) -- (C);
      \draw (b) -- (A);
      \draw (b) -- (B);
      \draw (b) -- (C);
    \end{tikzpicture}
    \caption{$K_{2,3}$.}
    \label{fig:K2x3}
  \end{subfigure}
  \hfill
  \begin{subfigure}{.30\linewidth}
    \centering
    \begin{tikzpicture}[scale=0.8]
      \coordinate (ca) at (0*72+18:0.7) {};
      \coordinate (cb) at (1*72+18:0.7) {};
      \coordinate (cc) at (2*72+18:0.7) {};
      \coordinate (cd) at (3*72+18:0.7) {};
      \coordinate (ce) at (4*72+18:0.7) {};

      \node[GraphNode] (a) at (ca) {};
      \node[GraphNode] (b) at (cb) {};
      \node[GraphNode] (c) at (cc) {};
      \node[GraphNode] (d) at (cd) {};
      \node[GraphNode] (e) at (ce) {};

      \node at ($(ca)+( 0.3,  0.0)$) {3};
      \node at ($(cb)+( 0.0,  0.3)$) {2};
      \node at ($(cc)+(-0.3,  0.0)$) {3};
      \node at ($(cd)+( 0.0, -0.3)$) {2};
      \node at ($(ce)+( 0.0, -0.3)$) {2};

      \draw (a) -- (b);
      \draw (a) -- (c);
      \draw (a) -- (e);
      \draw (b) -- (c);
      \draw (c) -- (d);
      \draw (d) -- (e);
    \end{tikzpicture}
    \caption{Graph whose degree multiset is the same as $K_{2,3}$ but is not isomorphic to $K_{2,3}$.}
    \label{fig:house}
  \end{subfigure}
  \hfill
  \begin{subfigure}{.40\linewidth}
    \centering
    \begin{tikzpicture}[scale=0.8]
      \coordinate (ca) at (0.5, 1) {};
      \coordinate (cb) at (1.5, 1) {};
      \coordinate (cA) at (0.0, 0) {};
      \coordinate (cB) at (1.0, 0) {};
      \coordinate (cC) at (2.0, 0) {};

      \node[GraphNode] (a) at (ca) {};
      \node[GraphNode] (b) at (cb) {};
      \node[GraphNode] (A) at (cA) {};
      \node[GraphNode] (B) at (cB) {};
      \node[GraphNode] (C) at (cC) {};

      \node at ($(ca)+( 0.0,  0.3)$) {(\myred{3}, \mygreen{0})};
      \node at ($(cb)+( 0.0,  0.3)$) {(\myred{0}, \mygreen{3})};
      \node at ($(cA)+( 0.0, -0.3)$) {(\myred{1}, \mygreen{1})};
      \node at ($(cB)+( 0.0, -0.3)$) {(\myred{1}, \mygreen{1})};
      \node at ($(cC)+( 0.0, -0.3)$) {(\myred{1}, \mygreen{1})};

      \draw[RedEdge] (a) to (A);
      \draw[RedEdge] (a) to (B);
      \draw[RedEdge] (a) to (C);
      \draw[GreenEdge] (b) to (A);
      \draw[GreenEdge] (b) to (B);
      \draw[GreenEdge] (b) to (C);

      \draw[RedEdge] (2.5, 0.75) -- (3.0, 0.75);
      \node[anchor = west] at (3.25, 0.75) {red};
      \draw[GreenEdge] (2.5, 0.25) -- (3.0, 0.25);
      \node[anchor = west] at (3.25, 0.25) {green};
    \end{tikzpicture}
    \caption{2-colorized graph of $K_{2,3}$.}
    \label{fig:K2x3_color}
  \end{subfigure}
  \caption{Graphs and a colored graph.}
  \label{fig:graphs}
\end{figure}

\subsection{Topological minors and characterizations of graphs}\label{sec:ftm}
% We introduce topological minors and related notion.
\emph{Subdividing} an edge $\edge{u}{v}$ of a graph $H$ means removing the edge $\edge{u}{v}$ from $H$, introducing a new vertex $w$, and adding new edges $\edge{u}{w}$ and $\edge{v}{w}$.
If a graph is obtained by subdividing each edge of $H$ arbitrary times, it is a \emph{subdivision} of $H$.
A graph $F$ is \emph{homeomorphic} to a graph $H$ if $F$ is isomorphic to some subdivision of $H$.\footnote{In another definition, $F$ is homeomorphic to $H$ if some subdivision of $F$ is isomorphic to some subdivision of $H$. However, we allow subdividing only for $H$ because $H$ is ``contracted enough'' when it is a forbidden topological minor, that is, $H$ does not contain redundant vertices with degree 2.}
If a graph $F$ is homeomorphic to a graph $H$, the original vertices of $H$ are the \emph{branch vertices} of $F$ and the other vertices are the \emph{subdividing vertices}.
Note that the degree of a branch vertex equals the original degree in $H$ while the degree of a subdividing vertex is 2. (The degree of a branch vertex can be 2 when its original degree in $H$ is 2.)
For graphs $G$ and $H$, $H$ is a \emph{topological minor} (TM) of $G$ if $G$ contains a subgraph homeomorphic to $H$.
A subgraph $G'$ of $G$ is a \emph{TM-embedding} of $H$ in $G$ if $G'$ is homeomorphic to $H$.
For families $\mathcal{G}$ and $\mathcal{H}$ of graphs, $\mathcal{G}$ is \emph{forbidden-TM-characterized} (FTM-characterized) by $\mathcal{H}$ when a graph $G$ belongs to $\mathcal{G}$ if and only if, for any subgraph $G'$ of $G$ and $H \in \mathcal{H}$, $G'$ is not homeomorphic to $H$.
% \begin{itemize}
%   \item A graph $G$ belongs to $\mathcal{G}$.
%   \item For any graph $H \in \mathcal{H}$, $G$ does not contain $H$ as a topological minor.
% \end{itemize}
For example, the family of planar graphs is FTM-characterized by $\set{K_5, K_{3, 3}}$~\cite{kuratowski1930sur}.
The same characterization goes to several graph classes (\tabref{tab:graph_classes}).

\subsection{Decision diagram (DD)}\label{sec:dd}
\begin{figure}[t]
    \centering
    \begin{tikzpicture}[scale=0.8,
      NonTerminalNode/.style={
        draw,
        circle,
      },
      TerminalNode/.style={
        draw,
        rectangle,
      },
      ZeroArc/.style={
        black,
        line width = 0.5mm,
        dashed,
        ->,
      },
      OneArc/.style={
        red,
        line width = 0.5mm,
        ->,
      },
      TwoArc/.style={
        green!50!black,
        line width = 0.5mm,
        dotted,
        ->,
      },
    ]

      \node[NonTerminalNode] (11) at ( 0.0,  1.5) {1};
      \node[NonTerminalNode] (21) at (-1.0,  0.5) {2};
      \node[NonTerminalNode] (22) at ( 0.0,  0.5) {2};
      \node[NonTerminalNode] (23) at ( 1.0,  0.5) {2};
      \node[NonTerminalNode] (31) at (-1.0, -0.5) {3};
      \node[NonTerminalNode] (32) at ( 0.0, -0.5) {3};
      \node[NonTerminalNode] (33) at ( 1.0, -0.5) {3};
      \node[TerminalNode]   (top) at ( 0.0, -1.5) {$\top$};

      \draw[ZeroArc] (11) to (21);
      \draw[OneArc]  (11) to (22);
      \draw[TwoArc]  (11) to (23);
      \draw[OneArc]  (21) to (31);
      \draw[TwoArc]  (21) to (32);
      \draw[ZeroArc] (22) to (31);
      \draw[TwoArc]  (22) to (33);
      \draw[ZeroArc] (23) to (32);
      \draw[OneArc]  (23) to (33);
      \draw[TwoArc]  (31) to (top);
      \draw[OneArc]  (32) to (top);
      \draw[ZeroArc] (33) to (top);

      \draw[ZeroArc] ( 1.75,  0.75) to ( 2.75,  0.75);
      \node[anchor = west] at ( 2.75,  0.75) {0-arc};

      \draw[OneArc]  ( 1.75,  0.0) to ( 2.75,  0.0);
      \node[anchor = west] at ( 2.75,  0.0) {1-arc};

      \draw[TwoArc]  ( 1.75, -0.75) to ( 2.75, -0.75);
      \node[anchor = west] at ( 2.75, -0.75) {2-arc};

    \end{tikzpicture}
    \caption{3-DD. A square is a terminal node and circles are non-terminal nodes. An integer in a circle is the label of the node. For simplicity, we omit $\bot$ and the arcs pointing at it.}
    \label{fig:dd}
\end{figure}
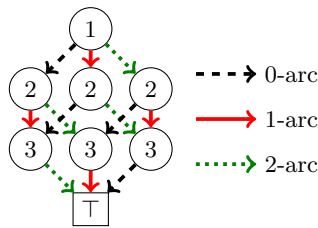

A \emph{$(c+1)$-decision diagram} ($(c+1)$-DD)~\cite{kawahara2019colorful} represents a family of $c$-colored subsets over a finite set.
We introduce it here briefly and describe it formally in \appref{app:dd_detail}.
Let $E = \set{e_1, \dots, e_m}$.
A $(c+1)$-DD over $E$ is a directed acyclic graph with the unique \emph{root node}\footnote{To avoid confusion, we use the terms ``node'' and ``arc'' for a $(c+1)$-DD and use ``vertex'' and ``edge'' for an input graph.} and the two \emph{terminal nodes} $\top$ and $\bot$.
Nodes other than the terminal nodes are the \emph{non-terminal nodes}.
Each non-terminal node has a \emph{label}, an integer in $[m]$, and $c+1$ arcs.
The arcs are the \emph{0-arc}, \emph{1-arc}, $\dots$, and \emph{$c$-arc}.
For each arc, the label of its tail is larger than its head if the tail is non-terminal.
In a $(c+1)$-DD, each path from the root node to $\top$ represents a $c$-colored subset over $E$.
For each path and $j \geq 1$, descending the $j$-arc of a non-terminal node with label $i$ corresponds to including $e_i$ in a colored subset as a color $j$ element.
Descending the $0$-arc corresponds to excluding $e_i$ from a colored subset.
The set of such paths corresponds to the family of $c$-colored subsets represented by the $(c+1)$-DD.
\figref{fig:dd} shows an example of a $3$-DD over a set $\set{e_1, e_2, e_3}$.
The 3-DD represents the family of the $2$-colored subsets each of which contains exactly one element for each color, that is,
$\{(\set{e_1}, \set{e_2})$,
$(\set{e_1}, \set{e_3})$,
$(\set{e_2}, \set{e_1})$,
$(\set{e_2}, \set{e_3})$,
$(\set{e_3}, \set{e_1})$,
$(\set{e_3}, \set{e_2})\}$.
In the rest of the paper, $\mathbf{Z}^{c+1}$ denotes a $(c+1)$-DD and $\DDToSet{\mathbf{Z}^{c+1}}$ denotes the family of $c$-colored subsets represented by $\mathbf{Z}^{c+1}$.
When $c = 1$, we omit the superscript from $\mathbf{Z}^{c+1}$ and write $\mathbf{Z}$, that is, $\mathbf{Z}$ is a 2-DD.

\subsection{Colorful frontier-based search (CFBS)}\label{sec:cfbs}
\emph{Colorful frontier-based search} (CFBS)~\cite{kawahara2019colorful} is a framework of algorithms to construct a DD representing the set of constrained subgraphs.
An important application of CFBS is the \emph{isomorphic subgraph enumeration problem}.
In the problem, given two graphs, a \emph{host graph} $G$ and a \emph{query graph} $H$, we output all subgraphs of $G$ that are isomorphic to $H$.
CFBS does not explicitly enumerate subgraphs but constructs a 2-DD representing the set of them.
By regarding a subgraph of $G$ as the subset of $E(G)$, 2-DD can represent a set of (edge-induced) subgraphs.

The crux of CFBS is to identify the query graph by a colored degree multiset.
\begin{definition}[profile]
  A multiset $M$ of $c$-colored degrees is a \emph{profile} of $H$ if the following are equivalent:
  \begin{itemize}
    \item A graph $F$ is isomorphic to $H$.
    \item There exists a $c$-colorized graph $F^c$ of $F$ that satisfies the constraint $\mathcal{C}_M$, where $\mathcal{C}_M$ is a function from $c$-colored graphs to $\set{0, 1}$ such that $\mathcal{C}_M(F^c) = 1 \Leftrightarrow \coloredDegreeSet{F^c} = M$.
  \end{itemize}
\end{definition}
For example, for $K_{2,3}$ in \figref{fig:K2x3}, $\set{(3, 0), (0, 3), (1, 1)^3}$ is a profile as shown in \figref{fig:K2x3_color}.
In the following, $\mathcal{G}^{c}(\mathcal{C})$ denotes the family of $c$-colored subgraphs of $G$ that satisfy the constraint $\mathcal{C}$.
$\mathcal{G}(H)$ denotes the family of subgraphs of $G$ that are isomorphic to $H$.
$\MDD{c+1}{\mathcal{C}}$ and $\DD{H}$ respectively denote the $(c+1)$-DD representing $\mathcal{G}^{c}(\mathcal{C})$ and the 2-DD representing $ \mathcal{G}(H)$.
The framework of algorithms for implicit isomorphic subgraph enumeration based on CFBS is as follows:
\begin{enumerate}
  \item Find a profile $M$ of $H$. Let $c$ be the number of the colors used in $M$.
  \item Construct $\mathbf{Z}^{c+1}(\mathcal{C}_M)$.
  \item By \emph{decolorizing} $\mathbf{Z}^{c+1}(\mathcal{C}_M)$, we obtain $\mathbf{Z}(H)$.
\end{enumerate}
Here, for a family $\mathcal{F}^c$ of $c$-colored subsets, its \emph{decolorization} is the family $\mathcal{F} = \setST{\bigcup_{i \in [c]} D_i}{(D_1, \dots, D_c) \in \mathcal{F}^c}$, that is, the family obtained by ignoring the colors from $\mathcal{F}^c$.
For DDs, the decolorization of the $(c+1)$-DD representing $\mathcal{F}^{c}$ is the 2-DD representing $\mathcal{F}$.
We can decolorize a DD by a recursive operation utilizing the recursive structure of the DD~\cite{kawahara2019colorful}.
% When $c = 1$, the algorithm corresponds to FBS and decolorization is not needed.
To construct a DD efficiently, CFBS uses dynamic programming.
The \emph{$i$-th frontier} $F_i$ is the set of vertices incident to both the edges in $\set{e_1, \dots, e_{i-1}}$ and $\set{e_i, \dots, e_m}$.
CFBS constructs a DD in a breadth-first manner from the root node and merges two nodes with the same label and states with respect to the frontier.
See \cite{kawahara2019colorful} for details.

CFBS can deal with every finite query graph $H$ in isomorphic subgraph enumeration if we use $|E(H)|$ colors~\cite{kawahara2019colorful}.
However, since the complexity of CFBS exponentially depends on the number of colors~\cite{kawahara2019colorful}, it is important to find a profile with as few colors as possible.
The following theorem states that, for every graph $H$, it suffices to use $\vertexCover{H}$ colors, where $\vertexCover{H}$ is the minimum size of vertex covers of $H$.
For a graph $H$, a subset $S \subseteq V(H)$ is a \emph{vertex cover} of $H$ if, for every edge $e \in E(H)$, at least one of its endpoints belongs to $S$.
A star is a graph isomorphic to $K_{1, a}$ for some positive integer $a$.

\begin{theorem}[\cite{horiyama2019decomposing}]\label{th:unigraph}
  Let $H$ be a graph and $H^c$ be a $c$-colored graph of $H$ such that, for every $i \in [c]$, the subgraph of $H^c$ induced by color-$i$ edges is isomorphic to a star.
  A graph $F$ is isomorphic to $H$ if and only if there exists a $c$-colored graph $F^{c}$ of $F$ such that $\coloredDegreeSet{F^{c}} = \coloredDegreeSet{H^c}$.
  It follows that, for every $H$, there is a profile using $\vertexCover{H}$-colored degrees.
\end{theorem}

\section{Algorithms}\label{sec:algo}
Proofs of the theorems in this section are shown in \appref{app:proof}.

\subsection{Implicit enumeration of TM-embeddings}\label{sec:general}
Given graphs $G$ and $H$, we propose an algorithm to construct the 2-DD $\mathbf{Z}(\widehat{H})$ representing the set of all TM-embeddings of $H$ in $G$.
In the following, $\mathcal{S}(H)$ denotes the family of subdivisions of $H$.
In an algorithm to enumerate isomorphic subgraphs based on existing CFBS, finding a profile of $H$ is essential.
However, as for TM-embeddings, since $\mathcal{S}(H)$ is an infinite set, it seems difficult to identify $\mathcal{S}(H)$ with a single colored degree multiset.
Therefore, we define an \emph{extended profile} for an infinite family $\mathcal{S}(H)$ by extending a profile for a graph $H$.

Let $\Delta^c = \setST{\deltaTwo{i}}{i \in [c]}$, where $\deltaTwo{i} = (\delta_1, \dots, \delta_c), \delta_i = 2, j \neq i \Rightarrow \delta_j = 0$.
\begin{definition}[extended profile]\label{def:ext_profile}
  Let $c$ be a positive integer.
  A multiset $M$ of $c$-colored degrees is an \emph{extended profile} of $\mathcal{S}(H)$ if the following are equivalent:
  \begin{itemize}
    \item A graph $F$ belongs to $\mathcal{S}(H)$.
    \item There exists a $c$-colorized graph $F^c$ of $F$ that satisfies the constraint $\mathcal{C}^{*}_M$, where $\mathcal{C}^{*}_M(F^c) = 1$ if and only if (a) $\coloredDegreeSet{F^c}$ is obtained by adding an arbitrary number of elements (allowing duplication) of $\Delta^c$ to $M$ and (b) each colored subgraph of $F^c$ is connected.
  \end{itemize}
\end{definition}

Our method of implicit TM-embedding enumeration is written as follows:
\begin{enumerate}
  \item Find an extended profile $M$ of $\mathcal{S}(H)$. Let $c$ be the number of colors in $M$.
  \item Construct $\MDD{c+1}{\mathcal{C}^{*}_M}$.
  \item By decolorizing $\MDD{c+1}{\mathcal{C}^{*}_M}$, we obtain $\DD{\widehat{H}}$.
\end{enumerate}
Decolorization in Step~3 can be done in the same way as existing CFBS.
To construct $\MDD{c+1}{\mathcal{C}^{*}_M}$ in Step~2, we add the constraints ``there are an arbitrary number of vertices whose degrees are in $\Delta^c$'' and ``each colored subgraph is connected'' to existing CFBS.
In fact, we propose an algorithm for a more general problem.
For a multiset $s$ and a set $t$ of $c$-colored degrees, let $\mathcal{C}_s^t$ be the corresponding constraint where we replace $\Delta^c$ by $t$ in the definition of $\mathcal{C}_s^*$.
In other words, for each $\delta \in s$, the number of $\delta$ in a subgraph must equal its multiplicity in $s$.
In addition, there can be an arbitrary number of vertices whose colored degrees are in $t$.
Given $s$ and $t$, we propose an algorithm to construct a DD $\MDD{c+1}{\mathcal{C}_s^t}$.
We show pseudocode in \appref{app:fbs}.

To assess the efficiency of algorithms based on CFBS, it is usual to analyze the \emph{width} of the output DD~\cite{nakahata2018IEICE,sekine1995tutte}.
The width of a DD is the maximum number of nodes with the same label.
It is a measure of both the size of the DD and the time complexity to construct the DD.
Let $w = \max_{i \in [m]} \size{F_i}$.
$\mathbb{N}$ denotes the set of non-negative integers.
For a $c$-tuple $\delta$, $\delta_i$ denotes the $i$-th element of $\delta$.
For $c$-tuples $\delta$ and $\gamma$ of integers, we define $\delta \leq \gamma$ if, for all $i \in [c]$, $\delta_i \leq \gamma_i$ holds.
When $\delta \leq \gamma$, we say that $\delta$ is \emph{dominated} by $\gamma$.
For a multiset $s$ and a set $t$ of $c$-tuples, $s \cup t$ denotes the set of tuples that appear in $s$ or $t$. (When a tuple is contained in $s$, its multiplicity is ignored in $s \cup t$.)
For a set $M$ of $c$-colored degrees, $\down{M}$ denotes the set of tuples in $\mathbb{N}^c$ that are dominated by a tuple in $M$, that is, $\down{M} = \setST{\chi \in \mathbb{N}^c}{\exists \delta \in M, \chi \leq \delta}$.
For $\delta \in s$, the multiplicity of $\delta$ in $s$ is denoted by $s(\delta)$.
\begin{theorem}\label{th:width}
  Given a multiset $s$ and a set $t$ of $c$-colored degrees, there is an algorithm to construct a DD $\MDD{c+1}{\mathcal{C}_s^t}$ with width
  \begin{equation}\label{eq:width}
    2^{\bigO{cw\log w}} \size{\down{s \cup t}}^{w} \prod_{\delta \in s} (s(\delta) + 1).
  \end{equation}
\end{theorem}
The theorem shows that the complexity of the algorithm mainly depends on $w$ and $c$.
Although $w$ is determined by the host graph $G$,
we can reduce the value of $c$ by finding an extended profile of the query graph $H$ using as few colors as possible.
We discuss how to find such an extended profile for general $H$ in the rest of this subsection and for specific graphs in the right column in \tabref{tab:graph_classes} in the next subsection.

We discuss the complexity for general $H$.
We assume that $H$ has at least two vertices.
Similarly to CFBS for isomorphic subgraph enumeration, we show that there is an extended profile for every graph $H$ using $\size{E(H)}$ and $\vertexCover{H}$ colors.
Recall that $\vertexCover{H}$ is the minimum size of vertex covers of $H$.
Although the latter is better in most cases, we show both theorems for comparison.

\begin{theorem}\label{th:edge}
  For a graph $H$, let $H^{|E(H)|}$ be a $|E(H)|$-colorized graph obtained by coloring the edges of $H$ with distinct colors.
  Then, $M = \coloredDegreeSet{H^{|E(H)|}}$ is an extended profile of $\subdivision{H}$.
  There is an algorithm to construct a DD $\MDD{|E(H)|+1}{\mathcal{C}_M^*}$ with width
  \begin{equation}\label{eq:width_edge_naive}
    2^{\bigO{|E(H)|w\log w} + |V(H)|} \left(2^{|E(H)|} + |E(H)|\right)^w.
  \end{equation}
\end{theorem}
\begin{theorem}\label{th:vertex}
  For a graph $H$, let $H^{\vertexCover{H}}$ be a $\vertexCover{H}$-colorized graph whose each colored subgraph is isomorphic to a star and the set of the centers is a minimum vertex cover of $H$.
  Then, $M = \coloredDegreeSet{H^{\vertexCover{H}}}$ is an extended profile of $\mathcal{S}(H)$.
  There is an algorithm to construct a DD $\MDD{\vertexCover{H}+1}{\mathcal{C}_M^*}$ with width
  \begin{equation}\label{eq:width_vertex}
    2^{\bigO{\vertexCover{H} w\log w} + |V(H)|} \left(2^{\vertexCover{H}}\vertexCover{H}\right)^w.
  \end{equation}
\end{theorem}

\subsection{Constraints for forbidden topological minors}\label{sec:profile}
We derive specific extended profiles for the subdivisions of the graphs in the right column of \tabref{tab:graph_classes}: complete graphs, complete bipartite graphs, and $K_4 - e$.
While the results for complete bipartite graphs and $K_4 - e$ follow directly from \theoref{th:vertex}, we can reduce one color for complete graphs.
In the following, we discuss complete bipartite graphs first, which is easier than complete graphs.

\begin{theorem}\label{th:complete_bipartite}
  Let $a, b\ (a \leq b)$ be positive integers.
  A multiset $M_{a, b} = M_{a, b}^1 \cup M_{a, b}^2$ consisting of $a$-colored degrees is an extended profile of $\mathcal{S}(K_{a, b})$, where
  \begin{equation}
    M_{a, b}^1 = \left\{ (\delta_1, \dots, \delta_a) \;\middle|\;
      \begin{array}{l}
        \exists i \in [a], \delta_i = b, \\
          j \neq i \Rightarrow \delta_j = 0
      \end{array}
    \right\}, \quad
    M_{a, b}^2 = \set{ (\underbrace{1, \dots, 1}_{a})^b }.
  \end{equation}
  There is an algorithm to construct a DD $\MDD{a+1}{\mathcal{C}_{M_{a,b}}^*}$ with width
  \begin{equation}\label{eq:width_complete_bipartite}
    2^{\bigO{aw\log w}} (2^a + ab)^w (b+1).
  \end{equation}
\end{theorem}
\figref{fig:complete_bipartite} shows a representation of a subdivision of $K_{3,3}$ based on \theoref{th:complete_bipartite}.

Next, we consider the subdivisions of complete graphs.
Since the size of a minimum vertex cover of $K_a$ is $a-1$, there exists an extended profile of $\mathcal{S}(K_a)$ with $a-1$ colors by \theoref{th:vertex}.
The extended profile is obtained by decomposing $K_{a}$ into $K_{1, 1}, K_{1, 2}, \dots,$ and $K_{1, a-1}$ and coloring the subgraphs with distinct colors.
In this coloring, if we color $K_{1, 2}$ with the same color as $K_{1, 1}$, the obtained subgraph is $K_{3}$.
We show that the colored degree multiset obtained from this coloring is also an extended profile of $\mathcal{S}(K_{a})$.

\begin{theorem}\label{th:complete}
  Let $a \geq 3$ be an integer.
  A multiset $M_{a-2} = M_{a-2}^1 \cup M_{a-2}^2$ consisting of $(a-2)$-colored degrees is an extended profile of $\mathcal{S}(K_a)$, where
  \begin{equation}
    \mkern-18mu M_{a-2}^1 = \set{ {(2, \underbrace{1, \dots, 1}_{a-3})}^3 },
    M_{a-2}^2 = \left\{(\delta_1, \dots, \delta_{a-2}) \;\middle|\;
      \begin{array}{l}
        \exists i \in \set{2, \dots, a-2}, \\
        j < i \Rightarrow \delta_j = 0, \\
        \delta_i = i+1, \\
        j > i \Rightarrow \delta_j = 1
      \end{array}
    \right\}.
  \end{equation}
  There is an algorithm to construct a DD $\MDD{a-1}{\mathcal{C}_{M_a}^*}$ with width
  \begin{equation}\label{eq:width_complete}
    2^{\bigO{aw\log w}} \left(3 \cdot 2^{a-2} - a\right)^w.
  \end{equation}
\end{theorem}
\figref{fig:complete} shows a representation of a subdivision of $K_{5}$ based on \theoref{th:complete}.

\begin{figure}[t]
  \centering
  \begin{subfigure}{.49\linewidth}
    \centering
    \begin{tikzpicture}[scale=0.8]
      \coordinate (ca) at (0, 2) {};
      \coordinate (cb) at (1, 2) {};
      \coordinate (cc) at (2, 2) {};
      \coordinate (cA) at (0, 0) {};
      \coordinate (cB) at (1, 0) {};
      \coordinate (cC) at (2, 0) {};

      \node[GraphNode] (a) at (ca) {};
      \node[GraphNode] (b) at (cb) {};
      \node[GraphNode] (c) at (cc) {};
      \node[GraphNode] (A) at (cA) {};
      \node[GraphNode] (B) at (cB) {};
      \node[GraphNode] (C) at (cC) {};

      \node at ($(ca)+(-0.5,  0.4)$) {(\myred{3}, \mygreen{0}, \myblue{0})};
      \node at ($(cb)+( 0.0,  0.4)$) {(\myred{0}, \mygreen{3}, \myblue{0})};
      \node at ($(cc)+( 0.5,  0.4)$) {(\myred{0}, \mygreen{0}, \myblue{3})};
      \node at ($(cA)+(-0.5, -0.4)$) {(\myred{1}, \mygreen{1}, \myblue{1})};
      \node at ($(cB)+( 0.0, -0.4)$) {(\myred{1}, \mygreen{1}, \myblue{1})};
      \node at ($(cC)+( 0.5, -0.4)$) {(\myred{1}, \mygreen{1}, \myblue{1})};

      \foreach \currentType/\currentCenter in {RedEdge/a, GreenEdge/b, BlueEdge/c}{
        \foreach \currentLeaf in {A, B, C}{
          \draw[\currentType] (\currentCenter) to (\currentLeaf);
        }
      }

      \node[SubdivisionNode] at ($(ca) !0.33! (cA)$) {};
      \node[SubdivisionNode] at ($(ca) !0.67! (cA)$) {};
      \node[SubdivisionNode] at ($(ca) !0.33! (cB)$) {};
      \node[SubdivisionNode] at ($(ca) !0.2! (cC)$) {};
      \node[SubdivisionNode] at ($(cb) !0.2! (cA)$) {};
      \node[SubdivisionNode] at ($(cb) !0.67! (cA)$) {};
      \node[SubdivisionNode] at ($(cb) !0.67! (cB)$) {};
      \node[SubdivisionNode] at ($(cb) !0.67! (cC)$) {};
      \node[SubdivisionNode] at ($(cc) !0.2! (cA)$) {};
      \node[SubdivisionNode] at ($(cc) !0.33! (cB)$) {};
      \node[SubdivisionNode] at ($(cc) !0.5! (cC)$) {};

      \coordinate (green_anchor) at ($(cC)+(1,1)$);
      \coordinate (red_anchor)   at ($(green_anchor)+(0,  0.5)$);
      \coordinate (blue_anchor)  at ($(green_anchor)+(0, -0.5)$);

      \foreach \currentType/\currentAnchor/\currentLabel in {RedEdge/red_anchor/red, GreenEdge/green_anchor/green, BlueEdge/blue_anchor/blue}{
        \draw[\currentType] (\currentAnchor) to ($(\currentAnchor)+(1,0)$);
        \node[anchor = west] at ($(\currentAnchor)+(1.25,0)$) {\currentLabel};
      }
    \end{tikzpicture}
    \caption{Representation of a subdivision of $K_{3,3}$ based on \theoref{th:complete_bipartite}.}
    \label{fig:complete_bipartite}
  \end{subfigure}
  \hfill
  \begin{subfigure}{.49\linewidth}
    \centering
    \begin{tikzpicture}[scale=0.8]
      \coordinate (cd) at (0*72+18:1);
      \coordinate (ce) at (1*72+18:1);
      \coordinate (ca) at (2*72+18:1);
      \coordinate (cb) at (3*72+18:1);
      \coordinate (cc) at (4*72+18:1);

      \node[GraphNode] (d) at (cd) {};
      \node[GraphNode] (e) at (ce) {};
      \node[GraphNode] (a) at (ca) {};
      \node[GraphNode] (b) at (cb) {};
      \node[GraphNode] (c) at (cc) {};

      \node at ($(0,0) !1.8! (cd)$) {(\myred{0}, \mygreen{3}, \myblue{1})};
      \node at ($(0,0) !1.3! (ce)$) {(\myred{0}, \mygreen{0}, \myblue{4})};
      \node at ($(0,0) !1.8! (ca)$) {(\myred{2}, \mygreen{1}, \myblue{1})};
      \node at ($(0,0) !1.4! (cb)$) {(\myred{2}, \mygreen{1}, \myblue{1})};
      \node at ($(0,0) !1.4! (cc)$) {(\myred{2}, \mygreen{1}, \myblue{1})};

      \node[SubdivisionNode] at ($(ca) !0.5! (cb)$) {};
      \node[SubdivisionNode] at ($(cb) !0.5! (cc)$) {};
      \node[SubdivisionNode] at ($(cc) !0.5! (cd)$) {};
      \node[SubdivisionNode] at ($(cd) !0.5! (ce)$) {};
      \node[SubdivisionNode] at ($(ce) !0.33! (ca)$) {};
      \node[SubdivisionNode] at ($(ce) !0.67! (ca)$) {};
      \node[SubdivisionNode] at ($(ca) !0.2! (cc)$) {};
      \node[SubdivisionNode] at ($(ca) !0.8! (cc)$) {};
      \node[SubdivisionNode] at ($(cb) !0.5! (cd)$) {};
      \node[SubdivisionNode] at ($(cb) !0.75! (ce)$) {};
      \node[SubdivisionNode] at ($(cc) !0.75! (ce)$) {};

      \draw[RedEdge] (a) to (b);
      \draw[RedEdge] (b) to (c);
      \draw[RedEdge] (c) to (a);
      \draw[GreenEdge] (d) to (a);
      \draw[GreenEdge] (d) to (b);
      \draw[GreenEdge] (d) to (c);
      \draw[BlueEdge] (e) to (a);
      \draw[BlueEdge] (e) to (b);
      \draw[BlueEdge] (e) to (c);
      \draw[BlueEdge] (e) to (d);
    \end{tikzpicture}
    \caption{Representation of a subdivision of $K_{5}$ based on \theoref{th:complete}.}
    \label{fig:complete}
  \end{subfigure}
  \caption{Representations of subdivisions of $K_{3,3}$ and $K_5$. In each figure, a filled and non-filled vertex represent a branch and subdividing vertex, respectively. A tuple beside a vertex means the colored degree of the vertex.}
  \label{fig:subdivisions}
\end{figure}
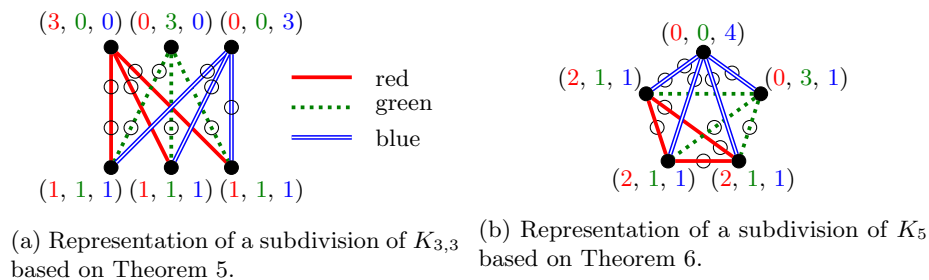

\subsection{Enumerating subgraphs having FTM-characterizations}\label{sec:FTM-DD}
We show how to implicitly enumerate subgraphs having FTM-characterization.
We combine DD operations with our algorithm to implicitly enumerate TM-embeddings.
\textsc{Union}~\cite{minato1993zero} is a function whose inputs are two 2-DDs $\mathbf{Z}_1$ and $\mathbf{Z}_2$ and output is the 2-DD representing $\DDToSet{\mathbf{Z}_1} \cup \DDToSet{\mathbf{Z}_2}$.
\textsc{NonSupset}~\cite{knuth2011art} is a function whose input is a 2-DD $\mathbf{Z}$ over a finite set $E$ and output is the 2-DD representing the family $\setST{A \subseteq 2^E}{\forall B \in \DDToSet{\mathbf{Z}}, A \not\supseteq B}$.
$\mathcal{G}(\widehat{H})$ denotes the set of subgraphs of $G$ that are homeomorphic to $H$ and $\mathbf{Z}(\widehat{H})$ denotes the 2-DD representing $\mathcal{G}(\widehat{H})$.
The following algorithm constructs the 2-DD representing the set of subgraphs of $G$ that is FTM-characterized by $\mathcal{H}$.
\begin{enumerate}
  \item Initialize a 2-DD $\ZDDSubd$ by the 2-DD representing the empty set.
  \item Choose an arbitrary graph $H$ from $\mathcal{H}$ and remove it from $\mathcal{H}$.
  \item Update $\ZDDSubd$ by ${\textsc{Union}}(\ZDDSubd, \mathbf{Z}(\widehat{H}))$.
  \item If $\mathcal{H}$ is not empty, go back to Step~2. If empty, go on to Step~5.
  \item We obtain the final 2-DD $\ZDDAns$ by ${\textsc{NonSupset}}(\ZDDSubd)$.
\end{enumerate}
For example, let us consider the case where we want to implicitly enumerate all planar subgraphs of $G$.
In this case, $\mathcal{H}$ is $\set{K_5, K_{3,3}}$.
We construct $\mathbf{Z}(\widehat{K_5})$ and $\mathbf{Z}(\widehat{K_{3,3}})$ and take their union, which is $\ZDDSubd$.
Now $\ZDDSubd$ represents the set of all subgraphs of $G$ that are homeomorphic to $K_5$ or $K_{3,3}$.
$\ZDDAns = {\textsc{NonSupset}}(\ZDDSubd)$ represents the family of all subgraphs of $G$ that is FTM-characterized by $\mathcal{H} = \set{K_5, K_{3,3}}$.
Therefore, $\ZDDAns$ represents the family of all planar subgraphs of $G$.
Other types of subgraphs such as outerplanar, series-parallel, and cactus subgraphs can be implicitly enumerated only by changing $\mathcal{H}$ according to \tabref{tab:graph_classes}.

\section{Computational experiments}\label{sec:experiment}
\subsection{Settings}\label{sec:exp_setting}
We conducted two experiments.
First, we compared several methods to enumerate planar subgraphs.
Second, we applied our framework to enumerate all types of subgraphs in \tabref{tab:graph_classes}.
We show the second results in \appref{app:exp_subgraphs}.
For input graphs, we used complete graphs $K_n$ and $3 \times b$ king graphs $X_{3, b}$ as synthetic data.
$X_{3, b}$ is a graph obtained by, to the $3 \times b$ grid graph, adding diagonal edges in all the cycles with length four.
As real data, we used Rome graph\footnote{\url{http://www.graphdrawing.org/data.html}}, which is often used in studies on graph drawing.
We implemented all the code in C++ and compiled them by g++5.4.0 with -O3 option.
To handle DDs, we used TdZdd~\cite{iwashita2013efficient} and SAPPORO\_BDD inside Graphillion~\cite{inoue2016graphillion}.
We used a machine with Intel Xenon E5-2637 v3 CPU and 1~TB RAM.
For each case, we set the timeout to one day.

\subsection{Comparing several methods to enumerate planar subgraphs}\label{sec:exp_methods}
We compare the following three methods for planar subgraph enumeration.
\begin{itemize}
  \item {\backtrack}: It explicitly enumerates subgraphs based on backtracking. The details are described in \appref{app:backtrack}.
  \item {\ddedge}: It implicitly enumerates subgraphs using DDs. It uses $|E(H)|$ colors based on \theoref{th:edge}. In other words, it uses ten colors for $\mathcal{S}(K_5)$ and nine colors for $\mathcal{S}(K_{3,3})$.
  \item {\ddvertex}: It implicitly enumerates subgraphs using DDs. It uses $\vertexCover{H}$ colors based on Theorems~\ref{th:vertex}--\ref{th:complete}. In other words, it uses three colors both for $\mathcal{S}(K_5)$ and $\mathcal{S}(K_{3,3})$.
\end{itemize}
As a subroutine of {\backtrack}, we used a planarity test in C++ Boost\footnote{\url{https://www.boost.org/doc/libs/1_71_0/libs/graph/doc/boyer_myrvold.html}}.
For fairness, {\backtrack} does not output solutions but only counts the number of solutions.
{\ddedge} and {\ddvertex} construct DDs representing the set of solutions.
Once a DD is constructed, we can count the number of solutions in linear time to the number of nodes in the DD~\cite{knuth2011art}.

\tabref{tab:exp_compare} shows the experimental results.
In all the cases, all the methods output the same number of solutions.
Among the three methods, {\ddvertex} ran fastest except for $K_6$.
{\backtrack} finished in a day only when the number of solutions is small (less than $10^9$).
Although {\ddedge} solved more instances than {\backtrack}, it ran out of memory when the size of input or the number of solutions grows.
In contrast, {\ddvertex} succeeded even for such instances.
For example, for $X_{3,4}$, {\ddvertex} is 122,544 and 187 times faster than {\backtrack} and {\ddedge}.
In addition, for $X_{3,500}$, {\ddvertex} succeeded in implicitly enumerating $7.95 \times 10^{1349}$ planar subgraphs only in $405.04$ seconds (less than seven minutes).
These results demonstrate the outstanding efficiency of {\ddvertex}.

\begin{table*}[t]
\centering
\caption{Experimental results. Each column shows the name of graphs, the number of vertices and edges, the running time of the three methods (in seconds), and the number of planar subgraphs. ``T/O'' and ``M/O'' mean time out and memory out, respectively.
``-'' means all the methods failed. The number of solutions for $K_{10}$ is from OEIS A066537, which is marked by `*'. We write the fastest time for each input graph in bold.}
\label{tab:exp_compare}
\scalebox{1.0}{
\begin{tabular}{l|rr|rrr|r}
graph & $|V|$ & $|E|$ & {\backtrack} & {\ddedge} & {\ddvertex} & \# solutions \\ \hline
$K_5$ & 5 & 10 & $\mathbf{<0.01}$ & 0.14  & $\mathbf{<0.01}$ & 1023 \\
$K_6$ & 6 & 15 & $\mathbf{<0.01}$ & 2.12  & 0.21  & 32071 \\
$K_7$ & 7 & 21 & 28.28  & 35.02  & \textbf{2.73}  & 1823707 \\
$K_8$ & 8 & 28 & 3113.64  & 620.84  & \textbf{66.34}  & 163947848 \\
$K_9$ & 9 & 36 & T/O & 15623.11  & \textbf{4694.41}  & 20402420291 \\
$K_{10}$ & 10 & 45 & T/O & T/O & T/O & *3209997749284 \\ \hline

$X_{3,4}$ & 12 & 29 & 11029.38  & 16.83  & \textbf{0.09}  & $ 5.33 \times 10^{8} $ \\
$X_{3,5}$ & 15 & 38 & T/O & 53.93  & \textbf{1.67}  & $ 2.70 \times 10^{11} $ \\
$X_{3,10}$ & 30 & 83 & T/O & 665.65  & \textbf{5.62}  & $ 8.93 \times 10^{24} $ \\
$X_{3,50}$ & 150 & 443 & T/O & M/O & \textbf{37.28}  & $ 1.29 \times 10^{133} $ \\
$X_{3,100}$ & 300 & 893 & T/O & M/O & \textbf{76.99}  & $ 2.03 \times 10^{268} $ \\
$X_{3,500}$ & 1500 & 4493 & T/O & M/O & \textbf{405.04}  & $7.95 \times 10^{1349} $ \\
$X_{3,1000}$ & 3000 & 8993 & T/O & M/O & M/O & - \\ \hline

$G_1$ (grafo1764.20) & 20 & 25 & 792.16  & 1.09  & \textbf{0.06}  & $ 3.35 \times 10^{7} $ \\
$G_2$ (grafo1760.28) & 28 & 39 & T/O & 96.81  & \textbf{3.76}  & $ 5.49 \times 10^{11} $ \\
$G_3$ (grafo10000.38) & 38 & 52 & T/O & 787.98 & \textbf{29.43}  & $ 4.50 \times 10^{15} $ \\
$G_4$ (grafo10008.42) & 42 & 61 & T/O & 38647.96 & \textbf{668.15}  & $ 2.30 \times 10^{18} $ \\
$G_5$ (grafo1378.46) & 46 & 62 & T/O & M/O & \textbf{796.48}  & $ 4.61 \times 10^{18} $ \\
$G_6$ (grafo1395.61) & 61 & 78 & T/O & M/O & \textbf{11992.12}  & $ 3.02 \times 10^{23} $ \\
$G_7$ (grafo5287.61) & 61 & 88 & T/O & M/O & M/O & - \\
$G_8$ (grafo9798.76) & 76 & 91 & T/O & M/O & \textbf{1709.64}  & $ 2.48 \times 10^{27} $ \\
$G_9$ (grafo10006.98) & 98 & 136 & T/O & M/O & M/O & - \\ \hline
\end{tabular}
}
\end{table*}

\section{Conclusion}\label{sec:conclusion}
Given graphs $G$ and $H$, we have proposed a method to implicitly enumerate topological-minor-embeddings of $H$ in $G$ using decision diagrams.
We also have shown a useful application of our method to enumerating subgraphs characterized by forbidden topological minors, that is, planar, outerplanar, series-parallel, and cactus subgraphs.
Computational experiments show that our method can find all planar subgraphs up to 122,544 times faster than a naive backtracking-based method and could solve more problems than the backtracking-based method.
Future work is extending our method from topological minors to general minors.

\bibliographystyle{splncs04}

\appendix

\section{Formal description of decision diagrams}\label{app:dd_detail}
Let $E$ be a finite set consisting of $m$ elements $e_1, \dots, e_m$.
A \emph{$(c+1)$-DD over $E$} is a rooted directed acyclic graph $\mathbf{Z}^{c+1} = (N, A, \ell)$, where $N$ is the set of nodes, $A \subseteq \setST{(\alpha, \beta)}{\alpha, \beta \in N, \alpha \neq \beta}$ is the set of (directed) arcs, and $\func{\ell}{N}{[m+1]}$ is a labeling function for nodes.
To avoid confusion, we use the terms ``node'' and ``arc'' for a $(c+1)$-DD and use ``vertex'' and ``edge'' for an input graph.
In addition, we represent a node of a $(c+1)$-DD using the Greek alphabet (e.g., $\alpha$, $\beta$) and a vertex of a graph using the English alphabet (e.g., $u$, $v$).
% For an arc $a = (\alpha, \beta)$, we call $\alpha$ and $\beta$ the \emph{head} and \emph{tail} of $a$, respectively.
There is an exactly one \emph{root node} $\rho(\mathbf{Z}^{c+1})$ in $N$ whose indegree is zero.
In addition, $N$ has exactly two \emph{terminal nodes} $\bot$ and $\top$ whose outdegrees are zero.
Nodes other than the terminal nodes are called \emph{non-terminal nodes}.
Each node has the \emph{label} $\lbl{\alpha} \in [m+1]$.
If $\alpha$ is a non-terminal node, its label is an integer in $[m]$.
If $\alpha$ is a terminal node, its label is $m+1$.
Each non-terminal node $\alpha$ has exactly $c+1$ arcs emanating from $\alpha$.
The arcs are called the \emph{$0$-arc}, \emph{$1$-arc}, $\dots$, and \emph{$c$-arc} of $\alpha$.
For an integer $j \in \set{0, \dots, c}$, $\alpha_j$ denotes the node pointed at by the $j$-arc of $\alpha$.
For each non-terminal node $\alpha$, $\lbl{\alpha_j} = \lbl{\alpha} + 1$ or $\lbl{\alpha_j} = m+1$ holds.
That is, $\alpha_j$ is the node whose label is one more than that of $\alpha$ or a terminal node.
It follows that $\mathbf{Z}^{c+1}$ does not contain a directed cycle.

In a $(c+1)$-DD, each directed path $\pi$ corresponds to a $c$-colored subset $\pathToSet{\pi} = (B_1, \dots, B_c)$ over $E$, where, for each $j \in [c]$,
\begin{equation}
  B_j = \setST{e_{\lbl{\beta}}}{\pi\ \mathrm{contains\ the}\ j\mathrm{\mbox{-}th\ arc\ of}\ \beta}.
\end{equation}
A $(c+1)$-DD $\mathbf{Z}^{c+1}$ represents the family
\begin{equation}
  \DDToSet{\mathbf{Z}^{c+1}} = \setST{\pathToSet{\pi}}{\pi\ \mathrm{is\ a\ directed\ path\ from}\ \rho(\mathbf{Z}^{c+1})\ \mathrm{to}\ \top}
\end{equation}
of $c$-colored subsets.
\figref{fig:dd} shows an example of a $3$-DD over a set $\set{e_1, e_2, e_3}$.
The 3-DD represents the family of the $2$-colored subsets that each contains exactly one element for each color, that is,
$\{(\set{e_1}, \set{e_2})$,
$(\set{e_1}, \set{e_3})$,
$(\set{e_2}, \set{e_1})$,
$(\set{e_2}, \set{e_3})$,
$(\set{e_3}, \set{e_1})$,
$(\set{e_3}, \set{e_2})\}$.
% A $(c+1)$-DD is said to be \emph{reduced} if, for all two distinct non-terminal nodes $\alpha$ and $\beta$, (a) $\lbl{\alpha} \neq \lbl{\beta}$ or (b) there exists an integer $j \in \set{0, \dots, c}$ such that $\alpha_j \neq \beta_j$.
% If a $(c+1)$-DD is not reduced, it can be reduced by repeatedly merging two nodes violating the above condition.
% The $3$-DD in \figref{fig:dd} is reduced.
% A $(c+1)$-DD can be regarded as a generalization of existing multi-valued decision diagrams (MDDs)~\cite{mcgeer1995fast} and 2-DD of binary decision diagrams (BDDs)~\cite{bryant1986graph} and ZDDs~\cite{minato1993zero}.
% A 2-DD is obtained by omitting the ``suppress rule'' from a BDD and a ZDD.

\section{Details of our algorithm}\label{app:fbs}
In this section, we discuss the detail of our algorithm.
Kawahara et al.~\cite{kawahara2019colorful} proposed an algorithm to construct a $(c+1)$-DD representing the set of $c$-colored subgraphs whose colored degree multiset coincides to a given multiset $s$.
We extend the algorithm to deal with connectivities of each colored subgraph and colored degrees that are allowed to exist arbitrary many.
The algorithm of Kawahara et al.\ defines configuration as a tuple of two arrays $\degr$ and $\dn$.
The first array $\degr$ is an array of colored degrees of the vertices in the frontier.
For a vertex $v$ and an integer $j' \in [c]$, $\degr[v]$ and $\degr[v][j']$ respectively denotes the colored degree of $v$ and the color-$j'$ degree of $v$.
The second array $\dn$ is an array of the numbers of fixed vertices having each colored degree in $s$.
For a colored degree $\delta \in s$, $\dn[\delta]$ denotes the number of fixed vertices having colored degree $\delta$.

We propose an algorithm to the following problem: Given $s$ and $t$, construct a DD $\MDD{c+1}{\mathcal{C}_s^t}$.
To solve the problem, we extend the configurations of the algorithm of Kawahara et al.\ by two new arrays $\comp$ and $\done$.
The array $\comp$ manages the connectivity of vertices in the frontier in each colored subgraph.
For color $j' \in [c]$, $\comp[j']$ is a partition of the frontier such that two vertices $u, v$ are connected if and only if they are contained in the same set in $\comp[j']$.
The another array $\done$ holds flags that which colored subgraphs are completed.
For color $j' \in [c]$, $\done[j'] = \true$ if and only if connected color-$j'$ subgraph is completed.
We associate $(\degr, \dn, \comp, \done)$ as a configuration with each node.
Pseudocode is shown in Algorithms~\ref{alg:construct} and \ref{alg:update}.
Algorithm~\ref{alg:construct} initializes the configurations of the root node and constructs a DD based on the framework of FBS.
The subroutine $\textsc{Child}$ is a function whose inputs are a node $\alpha$ and an integer $j \in \set{0, \dots, c}$ and output is $\alpha_j$.
It is shown in Algorithm~\ref{alg:update}.
The correctness of our extention immedietly follows from the algorithms for connectivity in ordinary FBS~\cite{yoshinaka2012finding}.

\begin{algorithm}[t]
\caption{Constructing the $(c+1)$-DD}
\label{alg:construct}

\Input{a multiset $s$ and a set $t$ of $c$-colored degrees}
\Output{a $(c+1)$-DD}
% \Output{a $(c+1)$-DD $\mathbf{Z}$}

let $\degr \gets [][]$ (an empty associative array), $\dn[\delta] \gets 0$ for all $\delta \in s$, $\comp[j'] \gets \set{\set{}}$ for all $j' \in [c]$, and $\done[j'] \gets \false$ for all $j' \in [c]$\;
construct a root node $\rho$ with a configuration $(\degr, \dn, \comp, \done)$\;
let $N_1 \gets \set{\rho}$, $N_i \gets \emptyset$ for $i \in \set{2, \dots, m}$ and $N_{m+1} \gets \set{\top, \bot}$\;
\For{$i = 1, \dots, m$}{
  \For{$\alpha \in N_i$}{
    \For{$j = 0, \dots, c$}{
      $\alpha_j \gets \textsc{Child}(\alpha, j)$\;
      \If{$\alpha_j \notin N_{i+1} \cup N_{m+1}$}{
        add a new node $\alpha_j$ with label $i+1$ to $N_{i+1}$\;
      }
      let $\alpha_j$ be the $j$-child of $\alpha$\;
    }
  }
}
% let $\mathbf{Z}$ be a $(c+1)-DD$ consisting of nodes of $N_1, \dots, N_{m+1}$\;
% reduce $\mathbf{Z}$\;
% \Return $\mathbf{Z}$\;
\Return the $(c+1)$-DD consisting of nodes of $N_1, \dots, N_{m+1}$\;
\end{algorithm}

\begin{algorithm}[t]
\caption{\textsc{Child}($\alpha, i, x$)}
\label{alg:update}

\Input{node $\alpha$ with configuration $(\degr, \dn, \comp, \done)$ and a child number $j$}
\Output{a node $\alpha_j$ that will be the $j$-th child of $\alpha$}

let $i \gets \ell(\alpha)$ and $\set{u_1, u_2} \gets e_i$\;
generate $\alpha_j$\;
let $\degr' \gets \degr$, $\dn' \gets \dn$, $\comp' \gets \comp$, and $\done' \gets \done$\;
\For{$k \in [2]$}{
  \If{$u_k \notin F_i$}{
    $\degr'[u_k] \gets (0, \dots, 0)$\;
    \lFor{$j' \in [c]$ such that $\done'[j'] = \false$}{$\comp[j'] \gets \comp[j'] \cup \set{\set{u_k}}$}
  }
}
\If{$j > 0$}{
  \lIf{$\done'[j] = \true$}{\Return $\bot$}
  \For{$k \in [2]$}{
    $\degr'[u_k][j] \gets \degr'[u_k][j] + 1$\;
    \lIf{for all $\delta \in s \cup t$, $\degr'[u_k][j] \not\leq \delta$}{\Return $\bot$}
  }
  for each $k \in [2]$, let $C(u_k)$ be the set containing $u_k$ in the current $\comp'[j]$\;
  \If{$C(u_1) \neq C(u_2)$}{$\comp'[j] \gets (\comp'[j] \setminus \set{C(u_1), C(u_2)}) \cup \set{C(u_1) \cup C(u_2)}$}
}

\For{$j' \in [c]$ such that $\done'[j'] = \false$}{
  % TODO: 孤立点でないものだけ考える
  let $L \gets \setST{C \in \comp'[j]}{C \cap F_{i+1} = \emptyset}$ and $S \gets \comp'[j] \setminus L$\;
  \lIf{$|L| > 1$}{\Return $\bot$}
  \ElseIf{$|L| = 1$}{
    \lIf{$|S| > 0$}{\Return $\bot$}
    \Else{
      $\done'[c] \gets \true$\;
      \If{for all $j' \in [c]$, $\done'[j'] = \true$}{
        \lIf{for all $\delta \in s$, $\dn'[\delta] = s(\delta)$}{\Return $\top$}
        \lElse{\Return $\bot$}
      }
      $\comp'[c] \gets \comp'[c] \setminus L$\;
    }
  }
}

\For{$k \in [2]$}{
  \If{$u_k \notin F_{i+1}$}{
    \If{$\degr'[u_k] \in s$}{
      $\dn'[\degr'[u_k]] \gets \dn'[\degr'[u_k]] + 1$\;
      \lIf{$\dn'[\degr'[u_k]] > s(\degr'[u_k])$}{\Return $\bot$}
    }
    \lElseIf{$\degr'[u_k] \notin t$}{\Return $\bot$}
    let $C(u_k)$ be the set containing $u_k$ in the current $\comp'[j]$\;
    $\comp' \gets (\comp' \setminus \set{C(u_k)}) \cup (\set{C(u_k) \setminus \set{u_k}})$\;
  }
}
\lIf{$i = m$}{\Return $\bot$}
let $(\degr', \dn', \comp', \done')$ be the configuration of $\alpha_j$\;
\Return $\alpha'_j$\;
\end{algorithm}

\section{Proofs omitted from \secref{sec:algo}}\label{app:proof}
We show the proofs omitted from \secref{sec:algo}.
For Theorems~\ref{th:edge}--\ref{th:complete}, there is two parts in the proofs: correctness of extended profiles and widths of the output DDs.
The titles of the paragraphs indicate them.

\subsection{Proof of \theoref{th:width}}
For a vertex $v$, the number of different values for $\degr[v]$ is at most $\size{\down{s \cup t}}$.
Since the frontier has at most $w$ vertices, the number of $\degr$ is at most $\size{\down{s \cup t}}^w$.
For a colored degree $\delta \in s$, the value of $\dn[\delta]$ is in $\set{0, \dots, s(\delta)}$.
Therefore, the number of different values for $\dn$ is $\prod_{\delta \in s} (s(\delta) + 1)$.
For each color $j' \in [c]$, $\comp$ maintains the partition of $w$ vertices in the frontier.
Thus, the number of different values for $\comp$ is $\left(2^{\bigO{w\log w}}\right)^c = 2^{\bigO{cw\log w}}$.
For each color $j' \in [c]$, $\done[j']$ is either $\true$ or $\false$, and thus the number of different values for $\done$ is $2^c$.
Taking a product of all the numbers,
\begin{align}
  &\size{\down{s \cup t}}^w \cdot \left(\prod_{\delta \in s} (s(\delta) + 1)\right) \cdot 2^{\bigO{cw\log w}} \cdot 2^c \\
  &= 2^{\bigO{cw\log w}} \size{\down{s \cup t}}^{w} \prod_{\delta \in s} (s(\delta) + 1).
\end{align}

\subsection{Proof of \theoref{th:edge}}
\paragraph*{Extended profile.}
Observe that a subdivision of a graph $H$ is obtained by replacing each of its edges by a path with length one or more.
Let us color the paths with distinct colors.
The colored degree multiset of the colorized graph, with the constraint ``each colored subgraph is connected,'' suffices to ensure that the graph is homeomorphic to $H$.
Each colored subgraph must be a path because there are two vertices with degree 1 and an arbitrary number of vertices with degree 2 and is connected.
In addition, two paths with different colors $i$ and $j$ share their endpoints if and only if there is a vertex whose both color-$i$ and color-$j$ degrees are both 1.
Therefore, for every graph $H$, we can identify $\subdivision{H}$ by the constraints with $|E(H)|$ colors.
% Let $F$ be a graph that is homeomorphic to $H$.
% Each edge in $F$ is associated with an edge in $H$.
% Let $E(H) = \set{e'_1, \dots, e'_m}$.
% If an edge $e$ in $F$ is associated with an edge $e'_i$ in $H$, we color $e$ by color $i$.
% The obtained $\vertexCover{H}$-colorized graph of $F$ satisfies (a) and (b) in \defref{def:ext_profile}.

\paragraph*{Width.}
We derive the width in \eqref{eq:width_edge_naive} from \eqref{eq:width}.
Now $c = |E(H)|$, $t = \Delta^{|E(H)|}$, and all the tuples in $s$ are dominated by $(1, \dots, 1)$ because at least one edge with each color incidents to a vertex.
Thus, $\down{s \cup t} = \down{\set{(1, \dots, 1}}) \cup \down{\Delta^{|E(H)|}} = \down{\set{(1, \dots, 1}}) \cup \Delta^{|E(H)|}$.
Since $\down{\set{(1, \dots, 1)}}$ and $\Delta^{|E(H)|}$ are disjoint,
$\size{\down{s \cup t}} = \size{\down{\set{(1, \dots, 1)}}} + \size{\Delta^{|E(H)|}} = 2^{|E(H)|} + |E(H)|$.
In addition, $\prod_{\delta \in s} (s(\delta) + 1) \leq 2^{\size{V(H)}}$ because there are at most $|V(H)|$ different tuples in $s$.
Based on the above discussion,
\begin{align}
  &2^{\bigO{cw\log w}} \size{\down{s \cup t}}^{w} \prod_{\delta \in s} (s(\delta) + 1) \\
  &\leq 2^{\bigO{|E(H)|w\log w}} \left(2^{|E(H)|} + |E(H)|\right)^w 2^{\size{V(H)}} \\
  &= 2^{\bigO{|E(H)|w\log w} + |V(H)|} \left(2^{|E(H)|} + |E(H)|\right)^w.
\end{align}
% When $H$ is connected, it suffices to maintain the whole connectivity.
% Thus we can drop the number $|E(H)|$ of colors from \eqref{eq:width_edge_naive} and obtain \eqref{eq:width_edge_efficient}.

\subsection{Proof of \theoref{th:vertex}}
\paragraph{Extended profile.}
Let $F$ be a graph that is homeomorphic to $H$.
Each edge in $F$ is associated with an edge in $H$.
If an edge $e$ in $F$ is associated with an edge $e'$ in $H$ and $e'$ has the color $i$ in $H^{\vertexCover{H}}$, we color $e$ with $i$ in $F$.
The obtained $\vertexCover{H}$-colorized graph of $F$ satisfies (a) and (b) in \defref{def:ext_profile}.

Let $F$ be a graph and $F^c$ be a $c$-colorized graph of $F$ such that it satisfies (a) and (b) in \defref{def:ext_profile}.
First, we show that, in $F^c$, each colored subgraph is homeomorphic to a star.
For each integer $i$ in $[c]$, let $M_i$ be the multiset of degrees of vertices in the color-$i$ subgraph of $F^c$.
By (a) in \defref{def:ext_profile}, the multiset $M_i$ satisfies one of the following:
\begin{enumerate}
  \item $M_i = \set{1^2, 2^*}$, where $2^{*}$ means there are an arbitrary number (perhaps zero) of vertices with degree 2, or
  \item $M_i = \set{x^1, 1^x, 2^*}$ for an integer $x \geq 3$.
\end{enumerate}
If $M_i = \set{1^2, 2^*}$, together with (b) in \defref{def:ext_profile}, the color-$i$ subgraph of $F^c$ is a path.
Thus, it is homeomorphic to $K_{1,1}$.
Note that it is also homeomorphic to $K_{1,2}$ if $M_i$ contains at least one vertex with degree $2$.
If $M_i = \set{x^1, 1^x, 2^*}$ for an integer $x \geq 3$, together with (b) in \defref{def:ext_profile}, the color-$i$ subgraph of $F^c$ is isomorphic to $K_{1,x}$.
For each color $i \in [c]$, we process the color-$i$ subgraph of $F^c$ as follows:
\begin{itemize}
  \item If $M_i = \set{x^1, 1^x, 2^*}$ for an integer $x \geq 3$, we smooth all the vertices with degree 2, where \emph{smoothing} a vertex $v$ with degree 2 means removing vertex $v$ and edges incident to it and connecting two vertices that were adjacent with $v$ by a new edge.
  \item If $M_i = \set{1^2, 2^*}$, we check $\coloredDegreeSet{H^{\vertexCover{H}}}$.
  If $\coloredDegreeSet{H^{\vertexCover{H}}}$ contains a vertex with color-$i$ degree 2 (note that there exists at most one such vertex in $\coloredDegreeSet{H^{\vertexCover{H}}}$), we smooth all the vertices but one with degree 2.
  If not, we smooth all the vertices with degree 2.
\end{itemize}
Let $I^c$ be the $c$-colorized graph obtained by the above procedure and $I$ be its underlying graph.
In $I^c$, each colored subgraph is isomorphic to a star and $\coloredDegreeSet{I^c} = \coloredDegreeSet{H^c}$.
Therefore, by \theoref{th:unigraph}, the graph $I$ is isomorphic to $H$.
Since $F$ is obtained by inserting smoothed vertices into edges in $I$, the graph $F$ is a subdivision of $I$.
It follows that $F$ is homeomorphic to $H$.

\paragraph*{Width.}
Now $c = \vertexCover{H}$ and $t = \Delta^{\vertexCover{H}}$.
Since each colored subgraph of $H^{\vertexCover{H}}$ is a star, every colored degree $\chi$ in $\coloredDegreeSet{H^{\vertexCover{H}}}$ satisfies that there exist at most one color $i$ such that $\chi_i$ exceeds 1.
Therefore,
\begin{equation}\label{eq:D_vertex_1}
  \down{s \cup t} \subseteq \setST{(\chi_1, \dots, \chi_{\vertexCover{H}}) \in \mathbb{N}^{\vertexCover{H}}}{\begin{array}{l}
    \exists i \in [\vertexCover{H}], \chi_i \leq |V(H)| - 1, \\
      j \neq i \Rightarrow \chi_j \leq 1
  \end{array}}
\end{equation}
Note that $|V(H)| - 1$ is an upper bound of the maximum degree of $H$.
From \eqref{eq:D_vertex_1}, we obtain $\size{\down{s \cup t}} \leq 2^{\vertexCover{H}}\vertexCover{H}$.
Combining the above with $\prod_{\delta \in s} (s(\delta) + 1) \leq 2^{\size{V(H)}}$, we obtain
\begin{align}
  &2^{\bigO{cw\log w}} \size{\down{s \cup t}}^{w} \prod_{\delta \in s} (s(\delta) + 1) \\
  &\leq 2^{\bigO{\vertexCover{H} w\log w}} \left(2^{\vertexCover{H}}\vertexCover{H}\right)^w 2^{\size{V(H)}} \\
  &= 2^{\bigO{\vertexCover{H} w\log w} + |V(H)|} \left(2^{\vertexCover{H}}\vertexCover{H}\right)^w.
\end{align}

% We derive the width in \eqref{eq:width_edge_naive} from \eqref{eq:width}.
% Now $c = |E(H)|$, $t = \Delta^{|E(H)|}$, and all the tuples in $s$ are dominated by $(1, \dots, 1)$ because at least one edge with each color incidents to a vertex.
% Thus, $D_{s \cup t} = \down{\set{(1, \dots, 1}}) \cup \down{\Delta^c} = \down{\set{(1, \dots, 1}}) \cup \Delta^{|E(H)|}$.
% Since $\down{\set{(1, \dots, 1}})$ and $\Delta^{|E(H)|}$ are disjoint,
% $\size{D_{s \cup t}} = \size{\down{\set{(1, \dots, 1}}} + \size{\Delta^{|E(H)|}} = 2^{|E(H)|} + |E(H)|$.
% In addition, $\prod_{\delta \in s} (s(\delta) + 1) \leq 2^{\size{V(H)}}$ because there are at most $|V(H)|$ different tuples in $s$.
% Based on the above discussion,
% \begin{align*}
%   &2^{\bigO{cw\log w}} \size{D_{s \cup t}}^{w} \prod_{\delta \in s} (s(\delta) + 1) \\
%   &\leq 2^{\bigO{cw\log w}} \left(2^{|E(H)|} + |E(H)|\right)^w 2^{\size{V(H)}} \\
%   &= 2^{\bigO{|E(H)|w\log w} + |V(H)|} \left(2^{|E(H)|} + |E(H)|\right)^w.
% \end{align*}

\subsection{Proof of \theoref{th:complete_bipartite}}
\paragraph*{Extended profile.}
Let $A$ and $B$ be the parts of $K_{a, b}\ (a \leq b)$ consisting of $a$ and $b$ vertices, respectively.
The set $A$ is a minimum vertex cover of $K_{a, b}$.
Let us decompose $K_{a, b}$ into $a$ stars such that their centers are $A$ and the leaves are $B$.
We color the stars with distinct colors from $[a]$.
In the colorized graph, the multisets of colored degrees of the vertices in $A$ and $B$ are $M_{a, b}^1$ and $M_{a, b}^2$, respectively.
By \theoref{th:vertex}, $M = M_{a, b}^1 \cup M_{a, b}^2$ is an extended profile of $\mathcal{S}(K_{a, b})$.

\paragraph*{Width.}
Now $c = a$, $s = M_{a,b}$, and $t = \Delta^a$.
Since $\down{s \cup t} = \down{M_{a,b}^1} \cup \down{M_{a,b}^2} \cup \down{\Delta^a} = \down{M_{a,b}^1} \cup \down{M_{a,b}^2}$, we obtain $\size{\down{s \cup t}} \leq \size{\down{M_{a,b}^1}} + \size{\down{M_{a,b}^2}} = ab + 2^a = 2^a + ab$.
Combining the above with $\prod_{\delta \in s} (s(\delta) + 1) = 2^a (b+1)$,
\begin{align}
  &2^{\bigO{cw\log w}} \size{\down{s \cup t}}^{w} \prod_{\delta \in s} (s(\delta) + 1) \\
  &\leq 2^{\bigO{aw\log w}} \left(2^a + ab\right)^w 2^a (b+1) \\
  &= 2^{\bigO{aw\log w}} \left(2^a + ab\right)^w.
\end{align}

\subsection{Proof of \theoref{th:complete}}
If a graph $F$ is homeomorphic to a graph $H$, the original vertices of $H$ are the \emph{branch vertices} of $F$ and the other vertices are the \emph{subdividing vertices}.
Note that the degree of a branch vertex equals the original degree in $H$ while the degree of a subdividing vertex is 2. (The degree of a branch vertex can be 2 when its original degree in $H$ is 2.)

\paragraph*{Extended profile.}
Let us decompose a subdivision $F$ of $K_{a}$ into subdivisions of $K_3, K_{1, 3}, \dots,$ and $K_{1, a-1}$ so that their centers and leaves are the branch vertices of $F$ and color them with distinct colors.
We denote the colorized graph by $J$.
$J$ is an $(a-2)$-colored graph and the multiset of colored degrees of the branch vertices in $J$ is $ M_{a-2} = M_{a-2}^1 \cup M_{a-2}^2$, where
$M_{a-2}^1$ and $M_{a-2}^2$ are the multisets of the colored degrees of (three arbitrarily chosen) branch vertices of a subdivision of $K_3$ and the centers of the subdivisions of the stars, respectively.
Therefore, $J$ satisfies the constraint $\mathcal{C}^{*}_M$, where $M$ is $M_{a-2}$.

We show that the converse is true by induction.
When $a = 3$, for a graph $F$, assume that there exists a $3-2 = 1$-colorized graph $F^1$ satisfying the constraint $\mathcal{C}^{*}_{M}$, where $M = M_1$.
Since $F_1$ has an arbitrary number of vertices of degree 2 and is connected, $F_1$ is a cycle, that is, a subdivision of $K_3$.
Next, for an integer $a \geq 3$, assume that ``For a graph $I$, if there exists an $(a-2)$-colored graph $I^{a-2}$ satisfying the constraint $\mathcal{C}^{*}_{M}$, where $M = M_{a-2}$, $I$ belongs to $\mathcal{S}(K_a)$'' is true.
For a graph $F$, assume that there exists an $(a-1)$-colored graph $F^{a-1}$ satisfying the constraint $\mathcal{C}^{*}_{M'}$, where $M' = M_{a-1}$.
Among the colored degree multiset of $F^{a-1}$, the part of colors from $1$ to $a-2$ is $M_{a-2}$ plus an arbitrary number of elements of $\Delta^{a-2}$.
Therefore, by the assumption, the underlying graph of the colored graph from color $1$ to $a-2$ in $F^{a-1}$ forms a subdivision of $K_{a}$.
The remaining part, the color-$(a-1)$ subgraph of $F^{a-1}$, has one vertex with degree $a$, $a$ vertices with degree 1, and an arbitrary number of vertices with degree 2 and is connected.
Therefore, the color-$(a-1)$ subgraph of $F^{a-1}$ forms a subdivision of $K_{1, a}$.
As for its center, its color-$(a-1)$ degree is $a$ and the degrees of the other colors are 0.
As for its leaves, their color-$(a-1)$ degrees are 1.
If the colored degree of a leaf belongs to $M_{a-1}^1$, it is a branch vertex of $K_3$.
Otherwise, $\delta_i = i+1$ implies that it is the center of a subdivision of $K_{1, i+1}$.
Therefore, $F^{a-1}$ is a graph obtained by merging the branch vertices of a subdivision of $K_{a}$ and the leaves of a subdivision of $K_{1, a}$.
It follows that the underlying graph of $F^{a-1}$ is homeomorphic to $K_{a+1}$.

\paragraph*{Width.}
Now $c = a - 2$, $s = M_{a-2}$, and $t = \Delta^{a-2}$.
For $\down{s \cup t}$, the following holds:
\begin{align}
  \down{s \cup t}
  &= \down{M_{a-2}^1} \cup \down{M_{a-2}^2} \cup \down{\Delta^{a-2}} \\
  &= \down{\set{(\underbrace{1, \dots, 1}_{a-2})}}
    \cup \setST{(\delta_1, \dots, \delta_{a-2})}{
      \begin{array}{l}
        \exists i \in [a-2], \\
        j < i \Rightarrow \delta_j = 0, \\
        2 \leq \delta_i \leq i+1, \\
        j > i \Rightarrow \delta_j = 1
      \end{array}} \\
  &= 2^{a-2} + \sum_{i=1}^{a-2} \left(i \cdot 2^{a-2-i}\right) \\
  &= 2^{a-2} + (2^{a-1} - a) \\
  &= 3 \cdot 2^{a-2} - a.
\end{align}
In addition, $\prod_{\delta \in s} (s(\delta) + 1) = (3+1) \cdot (1+1)^{a-3} = 2^{a-1}$ holds. Thus, we obtain
\begin{align}
  &2^{\bigO{cw\log w}} \size{\down{s \cup t}}^{w} \prod_{\delta \in s} (s(\delta) + 1) \\
  &\leq 2^{\bigO{(a-2)w\log w}} \left(3 \cdot 2^{a-2} - a\right)^w 2^{a-1} \\
  &= 2^{\bigO{aw\log w}} \left(3 \cdot 2^{a-2} - a\right)^w.
\end{align}

\section{Extended profile of $\mathcal{S}(K_4 - e)$}\label{sec:diamond}
Recall that $K_4 - e$ is the graph obtained by removing an arbitrary edge from $K_4$.

\begin{theorem}\label{th:diamond}
  A multiset $M = \set{(3, 0), (1, 2), (1, 1)^2}$ of 2-colored degrees is an extended profile of $\subdivision{K_4 - e}$.
  There is an algorithm to construct a DD representing $\MDD{3}{\mathcal{C}_M^*}$ with width $2^{\bigO{w\log w}}$.
\end{theorem}
\begin{proof}
  Let $K_4 - e = (\set{a, b, c, d}, \set{\edge{a}{b}, \edge{a}{c}, \edge{a}{d}, \edge{c}{b}, \edge{c}{d}})$.
  The set $\set{a, c}$ of vertices is a minimum vertex cover of the graph.
  We color the star with edges $\edge{a}{b}$, $\edge{a}{c}$, and $\edge{a}{d}$ by red (color 1) and that with $\edge{c}{b}$ and $\edge{c}{d}$ by green (color 2).
  The colored degree multiset of the colorized graph is $M$.
  By \theoref{th:vertex}, $M$ is an extended profile of $\subdivision{K_4 - e}$.
  We obtain the width as follows:
  \begin{align}
    &2^{\bigO{cw\log w}} \size{\down{s \cup t}}^{w} \prod_{\delta \in s} (s(\delta) + 1) \\
    &= 2^{\bigO{2w\log w}} \cdot 8^w \cdot (2 \cdot 2 \cdot 3) \\
    &= 2^{\bigO{w\log w}}.
  \end{align}
\end{proof}

\section{Details of backtracking-based method}\label{app:backtrack}
In this section, we show the details of an algorithm to explicitly enumerate planar subgraphs based on backtracking, which we used in \secref{sec:experiment}.
Pseudocode is given in Algorithm~\ref{alg:backtrack}.
Given a graph $G$, we first call ${\textsc{Main}(G)}$ (Line~1).
It calls a subfunction ${\textsc{Rec}}$.
Its inputs are a graph $G$, a subset of edges $S$, the index of the edge that should be processed next.
If $i = |E(G)| + 1$, we have found a solution, planar subgraph, and thus we output it (Line~3).
Otherwise, we guess whether $e_i$ is adopted for a solution of not.
We always call ${\textsc{Rec}}(G, S, i+1)$ because $S$ is a planar subgraph.
In contrast, we call ${\textsc{Rec}}(G, S \cup \set{e_i}, i+1)$ only if $G[S \cup \set{e_i}]$ is planar.
Since planar graphs are closed under taking subgraphs, the algorithm correctly outputs all the planar subgraphs.
The algorithm runs a planarity test $\bigO{|E(G)|}$ times for each solution.
Since a planarity test can be done in $\bigO{|V(G)|}$ time~\cite{hopcroft1974efficient}, the time complexity of the algorithm is $\bigO{N \cdot |E(G)| \cdot |V(G)|}$, where $N$ is the number of solutions.

% In the algorithm, we first fix the order of edges.
% Next, we process edges in the order and, for each edge, decide whether to adopt it for a solution (subgraph).
% When we adopt an edge, we check if the subgraph plus the edge is planar and, if not, we do pruning.
% When we have processed all the edges, we output the subgraph because it is a solution (planar subgraph).
% Since planar graphs are hereditary, that is, all subgraphs of planar graphs are also planar, the above algorithm outputs all the solutions without duplication.
% The algorithm runs a planarity test $\bigO{|E(G)|}$ times for each solution.
% Since a planarity test can be done in $\bigO{|V(G)|}$ time~\cite{hopcroft1974efficient}, the time complexity of the algorithm is $\bigO{N \cdot |E(G)| \cdot |V(G)|}$, where $N$ is the number of solutions.

\begin{algorithm}[t]
\caption{Enumerating planar subgraphs based on backtracking}
\label{alg:backtrack}

\Input{a graph $G$}
\Output{all planar subgraphs in $G$}
\SetKwProg{Def}{def}{:}{}

\Def{${\textsc{Main}}(G)$}{
  ${\textsc{Rec}}(G, \emptyset, 1)$\;
}

\Def{${\textsc{Rec}}(G, S, i)$}{
  \lIf{$i = |E(G)| + 1$}{output $S$}
  \Else{
    ${\textsc{Rec}}(G, S, i + 1)$\;
    \If{$G[S \cup \set{e_i}]$ is planar}{
      ${\textsc{Rec}}(G, S \cup \set{e_i}, i + 1)$;
    }
  }
}
\end{algorithm}

\section{Applying our framework to several types of subgraphs}\label{app:exp_subgraphs}
In this section, we apply our framework to enumerate all types of subgraphs in \tabref{tab:graph_classes}.
As stated in \secref{sec:FTM-DD}, to enumerate different type of sugraphs, it is enough to change $\mathcal{H}$, the set of forbidden topological minors.

Figures~\ref{fig:complete_graphwise}--\ref{fig:rome_graphwise} show the results.
We call an algorithm to enumerate planar subgraphs \textsc{Planar}, and so on.
The results for king graphs (\figref{fig:king_graphwise}) is easiest to understand.
We observe that \textsc{Planar} takes the most time because it uses three colors while the others two colors.
Among the algorithms using two colors, \textsc{Outerplanar} is most time-consuming because it needs two topological minors.
The reason why \textsc{Series-Parallel} run faster than \textsc{Cactus} is that $K_4$ has better ``regularity'' than $K_4 - e$, which makes the size of the output DD smaller.
Indeed, when $b = 500$, the size (number of nodes) of the DD constructed by \textsc{Series-Parallel} was $4,582,909$ while that by \textsc{Cactus} $7,289,225$.
The similar relation hold both for Figures~\ref{fig:complete} and \ref{fig:rome_graphwise}.
For rome graphs (\figref{fig:rome_graphwise}), the time for $G_8$ was smaller than $G_6$ although $G_8$ has more edges than $G_6$.
It is because $w$ of $G_8$ was smaller than that of $G_6$.

\begin{figure}[t]
    \centering
    \begin{subfigure}{.49\linewidth}
        \includegraphics[scale=0.3]{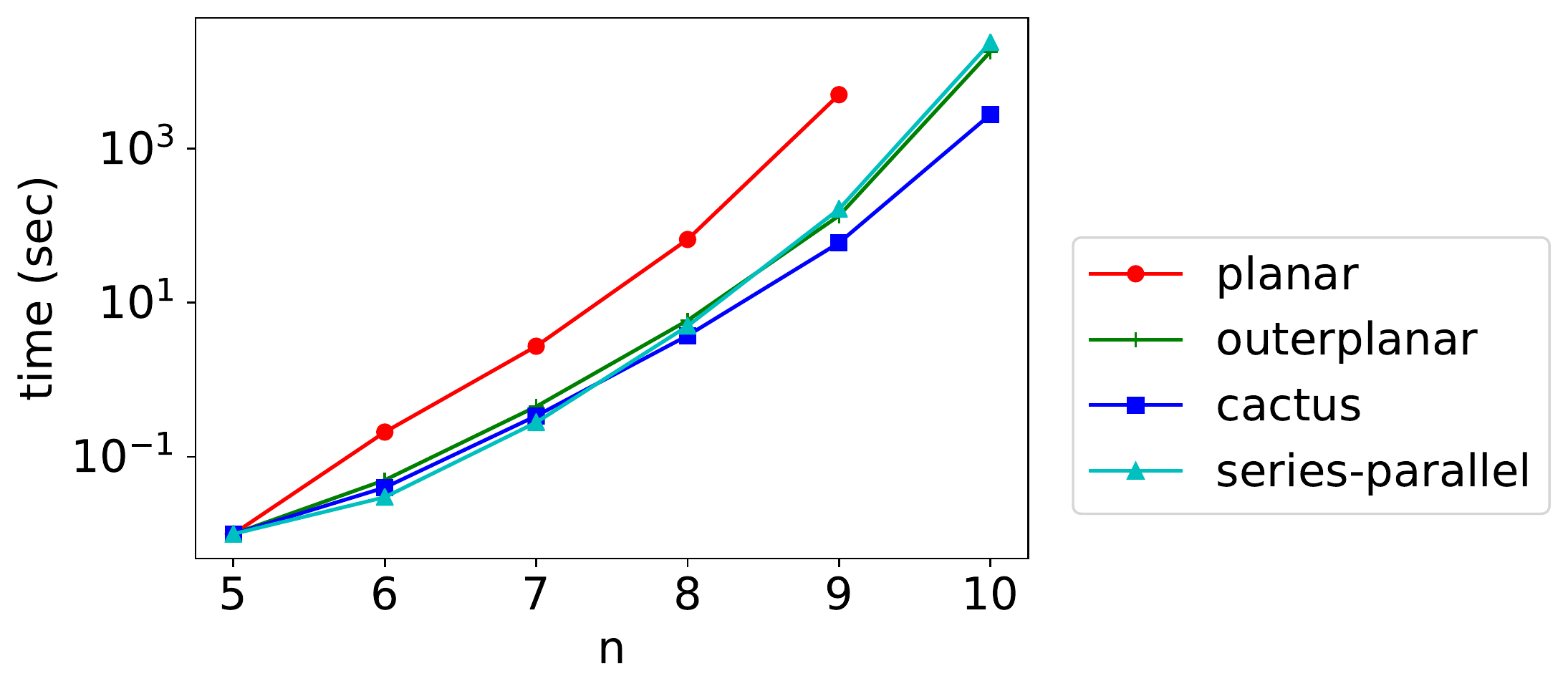}
        \caption{Complete graph $K_n$.}
        \label{fig:complete_graphwise}
    \end{subfigure}
    \hfill
    \begin{subfigure}{.49\linewidth}
        \centering
        \includegraphics[scale=0.3]{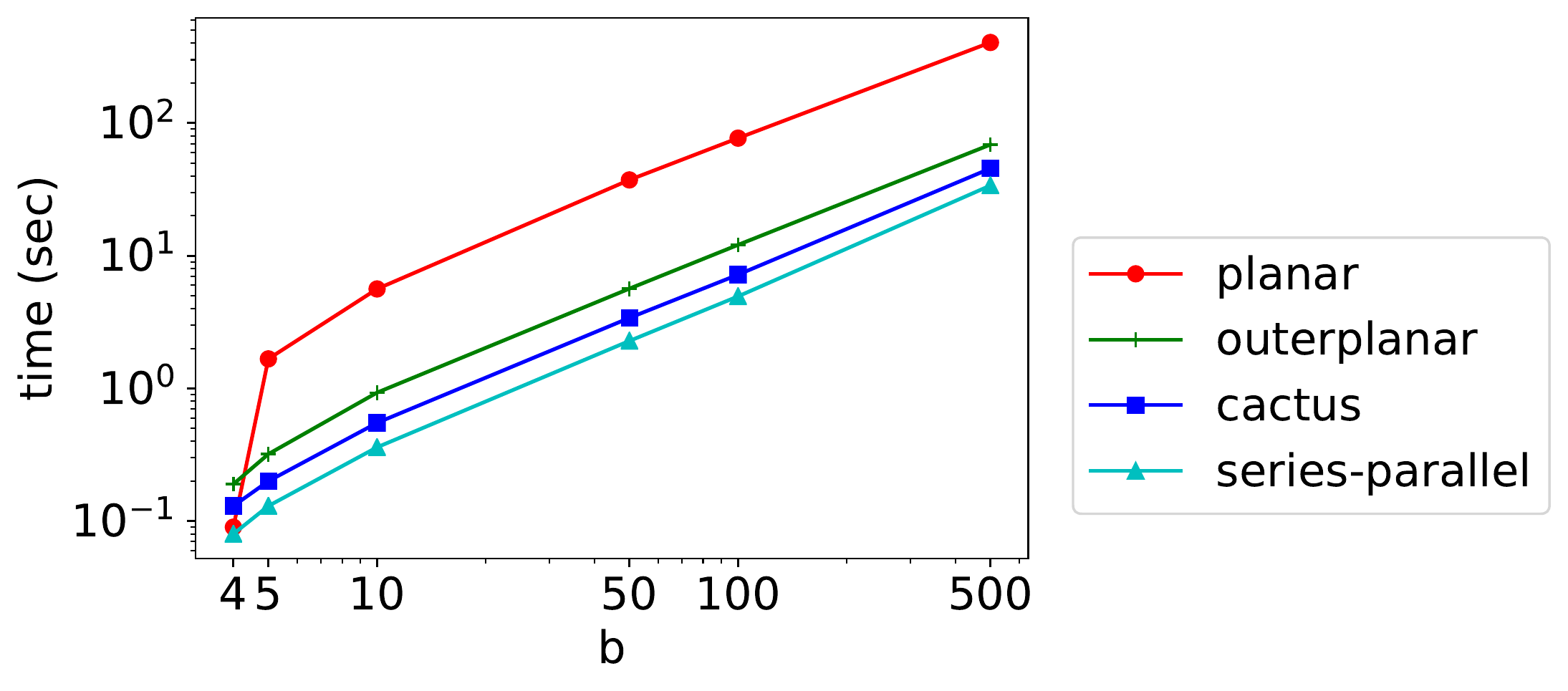}
        \caption{King graph $X_{3, b}$.}
        \label{fig:king_graphwise}
    \end{subfigure}
    \begin{subfigure}{.49\linewidth}
        \centering
        \includegraphics[scale=0.3]{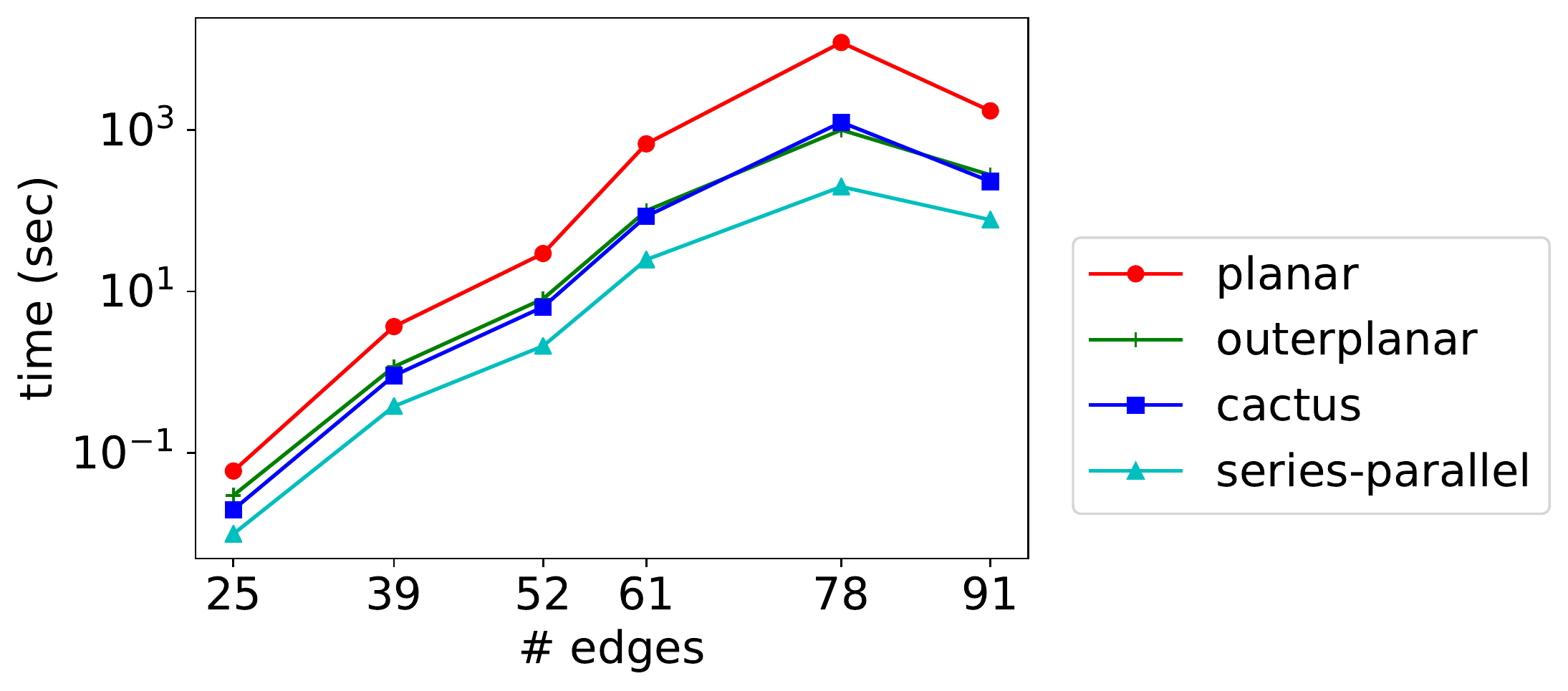}
        \caption{Rome graph.}
        \label{fig:rome_graphwise}
    \end{subfigure}
    \caption{Results of applying our framework to enumerate several types of subgraphs. For each figure, its horizontal axis shows the size of an input graph and vertical one running time (in seconds). Note that all the vertical axes and the horizontal axis of \figref{fig:king_graphwise} are logarithmic.}
\end{figure}

\end{document}

%% file: macro.tex
\usepackage[linesnumbered,ruled,vlined]{algorithm2e}
\SetKwInOut{Input}{input}
\SetKwInOut{Output}{output}
\usepackage{amsmath}
\usepackage{amssymb}
\usepackage{amsfonts}
\usepackage{cleveref}
\usepackage{stmaryrd}
\usepackage{url}
\usepackage[subrefformat=parens]{subcaption}
\usepackage{tikz}
\usetikzlibrary{
  positioning,
  calc,
}
\usepackage{lipsum}
\usepackage[shortlabels]{enumitem}

\tikzset{
  GraphNode/.style={
    draw,
    circle,
    fill=black,
    scale=0.6,
  },
  SubdivisionNode/.style = {
    draw,
    circle,
    scale=0.6,
  },
  RedEdge/.style={
    red,
    line width = 0.5mm,
  },
  GreenEdge/.style={
    green!50!black,
    line width = 0.5mm,
    dotted,
  },
  BlueEdge/.style={
    blue,
    line width = 0.2mm,
    double,
  },
}

\newcommand{\myred}[1]{\textcolor{red}{#1}}
\newcommand{\mygreen}[1]{\textcolor{green!50!black}{#1}}
\newcommand{\myblue}[1]{\textcolor{blue}{#1}}

\newcommand{\secref}[1]{Section~\ref{#1}}
\newcommand{\figref}[1]{Figure~\ref{#1}}
\newcommand{\tabref}[1]{Table~\ref{#1}}
\newcommand{\theoref}[1]{Theorem~\ref{#1}}
\newcommand{\defref}[1]{Definition~\ref{#1}}

\newcommand{\appref}[1]{Appendix~\ref{#1}}

\newcommand{\func}[3]{#1 \colon #2 \to #3}
\newcommand{\size}[1]{\left| #1 \right|}
\newcommand{\bigO}[1]{\mathcal{O}(#1)}

\newcommand{\set}[1]{\left\{ #1 \right\}}
\newcommand{\setST}[2]{\left\{ #1 \;\middle|\; #2 \right\}}

\newcommand{\tuple}[1]{\overrightarrow{#1}}

\newcommand{\edge}[2]{\set{#1, #2}}

\newcommand{\subdivision}[1]{\mathcal{S}(#1)}

\newcommand{\vertexCover}[1]{\tau(#1)}

\newcommand{\coloredDegreeSet}[1]{\mathrm{DS}(#1)}
\newcommand{\deltaTwo}[1]{\tuple{\delta \langle #1 \rangle}}

\newcommand{\pathToSet}[1]{\left\llbracket #1 \right\rrbracket}
\newcommand{\DDToSet}[1]{\pathToSet{#1}}
\newcommand{\DD}[1]{\mathbf{Z}(#1)}
\newcommand{\MDD}[2]{\mathbf{Z}^{#1}(#2)}
\newcommand{\lbl}[1]{\ell(#1)}

\newcommand{\ZDDSubd}{\mathbf{Z}_{\mathrm{subd}}}
\newcommand{\ZDDAns}{\mathbf{Z}_{\mathrm{ans}}}

\newcommand{\backtrack}{\textsc{Backtrack}}
\newcommand{\ddedge}{\textsc{DDEdge}}
\newcommand{\ddvertex}{\textsc{DDVertex}}

\newcommand{\degr}{\mathtt{deg}}
\newcommand{\dn}{\mathtt{dn}}
\newcommand{\comp}{\mathtt{comp}}
\newcommand{\done}{\mathtt{done}}
\newcommand{\true}{\mathit{True}}
\newcommand{\false}{\mathit{False}}

\newcommand{\down}[1]{\mathcal{D}\left(#1\right)}